\documentclass[aps,prl,twocolumn,superscriptaddress, amsmath, amssymb, amsfonts, nofootinbib]{revtex4-1}
\usepackage{amsthm}
\newtheorem{theorem}{Theorem}
\newtheorem{proposition}[theorem]{Proposition}

\newcommand{\commentout}[1]{}
\def\Re{\textnormal{Re}}

\def\Hess{\textnormal{Hess}}

\begin{document}

\title{Optimal robust quantum self-testing by binary nonlocal XOR games}

\author{Carl A.~Miller}
\email{carlmi@umich.edu}
\affiliation{Dept.~of Electrical Engineering and Computer Science, University of Michigan, Ann Arbor, MI  48109, USA}

\author{Yaoyun Shi}
\email{shiyy@umich.edu}
\affiliation{Dept.~of Electrical Engineering and Computer Science, University of Michigan, Ann Arbor, MI  48109, USA}

\date{\today}

\begin{abstract}
Self-testing a quantum device means verifying the existence of a certain quantum state
as well as the effect of the associated measurements based only on the statistics of the measurement outcomes.
Robust, i.e., error-tolerant, self-testing quantum devices
are critical building blocks for quantum cryptographic protocols that rely on imperfect or untrusted quantum devices.
We give a criterion which determines whether a given binary XOR game is robust self-testing with the asymptotically optimal error parameter.
As an application, we prove that the celebrated CHSH game is an optimally robust self-test.  We
also prove the same for
a family of tests recently proposed by Ac{\'i}n {\em et al.} ({\em PRL {\bf 108}:100402, 2012}) 
for random number generation, thus extending the benefit of the latter tests to allow imperfect or untrusted quantum devices.
\end{abstract}


\maketitle

Consider a quantum device with a classical input/output interface,
and suppose that the internal behavior of the device --- the quantum state inside and the measurements selected
by the classical input --- cannot be trusted to conform to a desired specification.
The device is said to be {\em self-testing}~\cite{my:1998}, if there exists
a self-test, i.e., a set of constraints on the input-output correlations, that once satisfied will guarantee
the accuracy to the specification.

The notion of quantum self-testing was explicitly formulated by Mayers and Yao~\cite{my:1998},
who pointed out its importance for quantum cryptography:  self-testing enables
quantum cryptographic protocols that rely on imperfect or untrusted quantum devices.
Such protocols were advanced in the recent thrust of research on 
{\em device-independent} quantum cryptography~\cite{AcinB:2007,
PironioA:2009, McKague:2009, pam:2010, Hanggi:2010, ColbeckK:2011, VaziraniV:2012}.

Multiple self-testing results are known.
Such results results are often based on nonlocal games.
Popescu and Rohrlich~\cite{pr:1992} proved that any state that achieves a maximal
violation of the CHSH inequality~\cite{CHSH} must be equivalent to a direct
sum of singlets. A self-testing result was proved for the GHZ paradox
by Colbeck~\cite{col:2006}.

In order for self-testing results to be practically useful, they must be robust --- that is, they must prove that an apparatus close to passing the test must be close to specification.
Robust self-testing was mentioned in \cite{my:2004}, and an early result was proved
in \cite{MagniezM:2006}.  The GHZ paradox is known to be a robust
self-test (as a special case of \cite{mck:2010}, also proved in \cite{millershi:2011}).  
Two recent preprints \cite{mys:2012, ruv:2012} prove that the CHSH inequality
is a robust self-test.

Existing proofs of self-testing are fairly lengthy and technical,
and appear specific to the underlying (class of) quantum states.
Also, there is some variation in the error terms afforded by these results.  
Some of the results on nonlocal games show that if the score
achieved is within $\epsilon$ of a passing score, the deviation
of the device from perfect behavior is no more that $C \sqrt{\epsilon}$.
This is easily seen to be the strongest
robustness possible, modulo the constant $C$.
For other results (e.g., in \cite{mck:2010,
mys:2012}) the error term is $C \epsilon^{1/4}$.
It important that these error terms be made as tight as possible.
In cryptogrphic protocols, 
a worse error term would require higher accuracy in observing the measurement
outcomes. Such more stringent requirement will then make the 
protocols fail with higher probability.

Most existing self-tests are based on binary nonlocal XOR games.
Those games are the most widely studied in the literature,
not only for the historical reason, but also because of the extremal
sensitivity of the XOR function on its input making it particularly
useful for contrasting classical and quantum games and for cryptographic applications.
In this paper, working within this important class,
we provide a simple criterion for robust self-testing.
The criterion determines whether a given game
satisfies robust self-testing with the optimal second-degree error term 
($C \sqrt{\epsilon}$).  
The criterion is fairly simple to check and
allows the proof of new optimal self-testing results.

In particular, we prove that the celebrated
CHSH game is a second-degree robust self-test.  
This result improves on the error term in \cite{mys:2012},
and was independently obtained in~\cite{ruv:2012}.
In addition, we show that a family of tests recently proposed
by Ac{\'i}n \textit{et al.} \cite{amp:2012} on randomness and
quantum correlations  are optimally self-testing.
The authors of \cite{amp:2012} characterized the qubit-devices that achieve the
optimal scores for those games, and argued that
these devices achieve more randomness 
than optimal devices for the standard CHSH inequality.
Our result that they are optimally robust shows that
the advantage of those games remains valid in the more practical settings
of imperfect or untrusted devices.

The starting point of our theory is the idea, first observed by Werner and Wolf~\cite{ww:2001},
that the optimal score
for a binary nonlocal XOR game can be expressed
as the maximum of a certain multivariable sinusoidal
function.  
In the present paper, we take the
idea a step further and show that the robust self-testing property 
can be checked using the local and global properties of this function.

We will begin with some definitions and then state
our main results.  The results are stated initially
for multiqubit states only, and then a higher-dimensional generalization
is given.
The proofs
are sketched here and written out in 
detail in the supplementary information.
We illustrate the usefulness of our theory through examples
and then conclude with open problems.

\paragraph*{Definitions.}

For our purposes, a \textit{binary nonlocal XOR game} is simply
a function $f \colon \{ 0, 1 \}^n \to \mathbb{R}$.  The function
$f$ describes a scoring rule for the game: if the input sequence
is $(i_1, i_2, \ldots, i_n)$, and the output sequence satisfies
$\oplus_k o_k = 0$, then the score is $f( i_1, i_2, \ldots, i_n )$;
if the input sequence is $(i_1, i_2, \ldots, i_n)$ and the output
sequence satisfies $\oplus_k o_k = 1$, then the score
is $- f ( i_1, i_2, \ldots, i_n )$.

To any nonlocal game $f$, let us associate a polynomial
$P_f \colon \mathbb{C}^n \to \mathbb{C}$ like so:
for any $n$-tuple $(\lambda_1 , \ldots, \lambda_n )$
of complex numbers, let $P_f \left( \lambda_1 , \ldots
, \lambda_n \right)$ be equal to 
\begin{eqnarray}
\sum_{(i_1, \ldots, i_n ) \in \{ 0, 1 \}^n }
f ( i_1, \ldots , i_n ) \lambda_1^{i_1} \lambda_2^{i_2}
\cdots \lambda_n^{i_n}.
\end{eqnarray}
For example, if $g$ is the CHSH
game ($g( 1, 1 ) = -1$, $g(0, 0) = g( 0, 1) = g(1, 0) = 1$) then
\begin{eqnarray}
\label{chshpolynomial}
P_g  = 1 +  \lambda_1
+ \lambda_2 - \lambda_1 \lambda_2.
\end{eqnarray}

Additionally, for any binary nonlocal XOR game
$f \colon \{ 0, 1 \}^n \to \mathbb{R}$, and any real numbers $\theta_0 , \theta_1 ,
\ldots , \theta_n$, let $Z_f ( \theta_0, \ldots , \theta_n)$
denote the quantity
\begin{eqnarray}
\sum_{(i_k) \in \{ 0 , 1 \}^n } f ( i_1, \ldots, i_n )
\cos \left( \theta_0 + \sum_k i_k \theta_k \right).
\end{eqnarray}
Thus, 
\begin{eqnarray}
\label{chshfunc}
Z_g ( \theta_0 , \theta_1 , \theta_2) = \cos ( \theta_0 ) + \cos ( \theta_0 + \theta_1 )
\\ \nonumber + \cos ( \theta_0 + \theta_2) - \cos (\theta_0 + \theta_1 + \theta_2).
\end{eqnarray}
Note that the function $Z_f$ is $2 \pi$-periodic in every variable, and it
satisfies $Z_f ( \theta_0, \ldots, \theta_n ) = Z_f ( -\theta_0, \ldots, -\theta_n )$.

The two quantities $P_f$ and $Z_f$ are related by the following identity.
\begin{eqnarray}
\label{identity}
Z_f ( \theta_0, \ldots , \theta_n) & = & \Re [ e^{i \theta_0 } P_f ( e^{i \theta_1 } , \ldots , e^{i \theta_n } ) ]. 
\end{eqnarray}
Note also that
\begin{eqnarray}
\left| P_f \left( e^{i \theta_1 } , \ldots, e^{i \theta_n } \right) \right| 
& = & \max_{t \in [-\pi , \pi ] } Z_f ( t, \theta_1,  \ldots, \theta_n ).
\end{eqnarray}

\paragraph*{Quantum strategies.}

For our purposes, a \textit{quantum strategy} for a binary $n$-player nonlocal
game is a pure state
\begin{eqnarray}
\left| \psi \right>  \in \mathcal{Q}_1 \otimes \mathcal{Q}_2 \otimes
\ldots \otimes \mathcal{Q}_n,
\end{eqnarray}
where each $\mathcal{Q}_j$ is a finite-dimensional Hilbert space,
together with two projective measurements
\begin{eqnarray}
\left\{ P^{(0, +)}_j , P^{(0, -)}_j \right\} ,
\left\{ P^{(1, +)}_j , P^{(1, -)}_j \right\}
\end{eqnarray}
on the space $\mathcal{Q}_j$.   These measurements
can be more compactly expressed as Hermitian operators:
\begin{eqnarray}
M_j^{(0)} & := & P^{(0, +)}_j  - P^{(0, -)}_j \\
M_j^{(1)} & := & P^{(1, +)}_j  - P^{(1, -)}_j 
\end{eqnarray}
The \textit{score} for such a strategy is
is the quantity $\left< \psi \right| \mathbf{M} \left| \psi \right>$,
where 
\begin{eqnarray}
\label{tensoroperator}
\mathbf{M} := \sum_{(i_k)} f( i_1, \ldots, i_n ) M_1^{(i_1)}  \otimes
\ldots \otimes M_n^{(i_n)}.
\end{eqnarray}

Let us use the term \textit{qubit strategy} to refer to a strategy whose
Hilbert spaces $\mathcal{Q}_j$ are all copies of $\mathbb{C}^2$ and whose
projection operators $P^{(i, *)}_j$ are all one-dimensional projectors.

For any nonlocal game $f$, let $q_f$ denote the highest possible score
for $f$ that can be achieved by a qubit strategy.  This
quantity has a relationship to the functions $Z_f$ and $P_f$
which was proved in \cite{ww:2001}.  For the benefit of
our exposition, we include a proof here.

\begin{proposition}[Werner and Wolf~\cite{ww:2001}]
\label{maxformulaprop}
Let $f \colon \{ 0, 1 \}^n \to \mathbb{R}$ be a nonlocal binary XOR
game.  Then,
\begin{eqnarray}
\label{optformula1}
q_f = \max_{\left| \lambda_1 \right| = \ldots = \left|  \lambda_n \right| = 1}
\left| P_f ( \lambda_1 , \ldots, \lambda_n ) \right| 
\end{eqnarray}
and
\begin{eqnarray}
\label{optformula2}
q_f = \max_{\theta_0 , \ldots, \theta_n \in [- \pi , \pi]} Z_f ( \theta_0 , \ldots
\theta_n ).
\end{eqnarray}
\end{proposition}

\begin{proof}
Let $(  \psi , \{ \{ M_j^{(0)} , M_j^{(1)} \} \}_j )$
be a qubit strategy for $f$ whose measurement operators
have the form
\begin{eqnarray}
\label{canonicalform}
M_j^{(0)} = \left[ \begin{array}{cc}
0 & 1 \\
1 & 0 \end{array} \right],
\hskip0.3in  M_j^{(1)} = \left[ \begin{array}{cc}
0 & e^{i \theta_j} \\
e^{-i \theta_j} & 0 \end{array} \right]
\end{eqnarray}
over the computational basis $\{ \left| 0 \right> , \left| 1 \right> \}$,
with $\theta_0 , \ldots, \theta_n \in [ - \pi , \pi ]$.
Any qubit strategy is equivalent under local unitary transformations
to such a strategy, so it suffices to compute the maximum score
achieved by strategies in this form.

The score for this quantum strategy is clearly bounded 
by the operator norm of $\mathbf{M}$, where $\mathbf{M}$
is the operator defined in terms of $\{ \{ M_j^{(0)} , M_j^{(1)} \} \}$ by
the tensor product expression (\ref{tensoroperator}).
The operator $\mathbf{M}$ is on a Hilbert space which has basis
$\left\{ \left| a_1 a_2 \ldots a_n \right> \mid a_i \in \{ 0, 1 \} \right\}$.
If we take the elements of this basis in lexicographical order, the 
matrix for $\mathbf{M}$ is a reverse-diagonal matrix:
\begin{eqnarray}
\left[ \begin{array}{ccccc}
0 & 0 & \ldots & 0 & * \\
0 & 0 & \ldots & * & 0 \\
\vdots & \vdots &  & \vdots & \vdots \\
0 & * & \ldots & 0 & 0 \\
* & 0 & \ldots & 0 & 0
\end{array} \right]
\end{eqnarray}
The entries along the reverse diagonal are given
by the expressions
\begin{eqnarray}
\label{entries}
P_f \left( e^{i (-1)^{a_1} \theta_1} , \ldots ,
e^{i (-1)^{a_n} \theta_n} \right)
\end{eqnarray}
for $(a_k ) \in \{ 0, 1 \}^n$.  

For any reverse-diagonal Hermitian
matrix whose reverse-diagonal entries
are $(z_1, z_2, \ldots, z_n, \overline{z_n} , 
\ldots, \overline{z_2} , \overline{z_1 } )$, the eigenvalues
are precisely 
$\pm \left| z_1 \right| ,\pm \left| z_2 \right| , 
\ldots , \pm\left| z_n \right|$. Therefore,
the operator norm of $\mathbf{M}$ is
\begin{eqnarray}
\label{opnormformula}
\max_{(a_i) \in \{ 0, 1 \}^n} \left|  P_f \left( e^{i (-1)^{a_1} \theta_1} , \ldots
e^{i (-1)^{a_n} \theta_n} \right) \right|.
\end{eqnarray}
Formula (\ref{optformula1}) follows.  Formula (\ref{optformula2})
follows also via equality (\ref{identity}).
\end{proof}

\paragraph*{Self-testing.}

Let $f$ be a binary nonlocal XOR game.  Let us say that $f$
is a \textit{self-test} if the following condition holds:
\begin{itemize}
\item[(*)] There is a single optimal qubit strategy $( \phi , \{ M_j^{(0)} ,
M_j^{(1)} \} \}_j )$ such that for any other optimal qubit strategy $( \psi , \{ N_j^{(0)} ,
N_j^{(1)} \} \}_j )$, there exist unitary
matrices $U_j \colon \mathbb{C}^2 \to \mathbb{C}^2$
such that
\begin{eqnarray}
(U_1 \otimes U_2 \otimes \ldots \otimes U_n )
\psi = \phi
\end{eqnarray}
and
$U_j N_j^{(i)} U_j^\dagger = M_j^{(i)}$
for all $i \in \{ 0 , 1\}$, $j \in \{ 1, 2, \ldots, n \}$.
\end{itemize}

\begin{proposition}
\label{selftestprop}
Let $f \colon \{ 0, 1 \}^n \to \mathbb{R}$ be a nonlocal binary XOR
game.  Then $f$ is a self-test if and only if the following two conditions
hold.
\begin{itemize}
\item[(A)] There is a maximum $(\alpha_0, \alpha_1, \ldots , \alpha_n)$
for $Z_f$ such that none of $\alpha_1, \alpha_2,  \ldots, \alpha_n$ is a multiple of $\pi$.

\item[(B)] Every other maximum of $Z_f$ is congruent modulo $2 \pi$ 
to either $(\alpha_0, \ldots, \alpha_n)$ or $(-\alpha_0, \ldots, - \alpha_n)$.
\end{itemize}
\end{proposition}

\begin{proof}
Suppose that conditions (A) and (B) both hold.  Without loss of generality,
we may assume $\alpha_0, \ldots , \alpha_n \in [ -\pi, \pi]$. Consider the qubit
strategy given by
\begin{eqnarray}
\label{optstrat1}
\phi = \frac{1}{\sqrt{2}} \left( \left| 00 \ldots 0 \right>
+ e^{i \alpha_0 }
\left| 11 \ldots 1 \right> \right).
\end{eqnarray}
and
\begin{eqnarray}
\label{optstrat2}
M_j^{(0)} = \left[ \begin{array}{cc}
0 & 1 \\
1 & 0 \end{array} \right],
\hskip0.3in  M_j^{(1)} = \left[ \begin{array}{cc}
0 & e^{i \alpha_j} \\
e^{-i \alpha_j} & 0 \end{array} \right].
\end{eqnarray}
which we will denote by $T ( \alpha_0, \alpha_1, \ldots, \alpha_n )$.
This strategy achieves the optimal score $q_f = Z_f ( \alpha_0, \alpha_1, \ldots, \alpha_n) =
\left| P_f ( \alpha_1 , \ldots, \alpha_n ) \right|$.

Suppose that $(\psi, \{ \{ N_j^{(0)} , N_j^{(1)} \} \} )$ is another
optimal qubit strategy.  After a unitary change of basis we may assume that
\begin{eqnarray}
N_j^{(0)} = \left[ \begin{array}{cc}
0 & 1 \\
1 & 0 \end{array} \right],
\hskip0.3in  N_j^{(1)} = \left[ \begin{array}{cc}
0 & e^{i \theta_j} \\
e^{-i \theta_j} & 0 \end{array} \right].
\end{eqnarray}
with $\theta_i \in [-\pi , \pi ]$,
and we may assume additionally that $(\alpha_i)$ and $(\theta_i)$ lie in the same quadrant (i.e.,
$\alpha_i \theta_i \geq 0$ for all $i \in \{ 1, \ldots, n \}$).

Consider the operator
\begin{eqnarray}
\mathbf{N} :=  \sum_{(i_k)} f( i_1, \ldots, i_n ) N_1^{(i_1)}  \otimes
\ldots \otimes N_n^{(i_n)}.
\end{eqnarray}
Among the reverse-diagonal entries of this operator (recall formula
(\ref{entries})) there must be a pair of conjugate entries that have absolute
value equal to $q_f$.  This
is possible only if $(\theta_1, \ldots, \theta_n) = (\alpha_1, \ldots, \alpha_n)$.
Moreover, there cannot be more than two entries having absolute value
$q_f$ (since otherwise (B) would be violated) and so we conclude
that the $q_f$-eigenspace of $\mathbf{N}$ is one-dimensional
and is spanned by $\phi$.  Thus $\psi$ is a scalar multiple
of $\phi$.  We conclude that $f$ is a self-test.

Now we show that conditions (A) and (B) are necessary.
Suppose first that (A) does not hold.  Then, there is a maximum
$(\alpha_0 , \ldots , \alpha_n )$ for $Z_f$ such that for some
$j \in \{ 1, \ldots, n \}$, $\alpha_j$ is a multiple of $\pi$.  Assume
without loss of generality that $j = n$.  Then the strategy described by
(\ref{optstrat1})--(\ref{optstrat2}) is optimal, and the strategy
where $\phi$ is replaced by the $n$-qubit state
\begin{eqnarray*}
\phi' = \left[ \frac{1}{\sqrt{2}} \left( \left| 00 \ldots 0 \right>
+ e^{i \alpha_0 }
\left| 11 \ldots 1 \right> \right) \right] \otimes \left(
\frac{ \left| 0 \right> + \left| 1 \right> }{\sqrt{2}} \right)
\end{eqnarray*}
is also optimal.  The states $\phi$ and $\phi'$ cannot
be related by local unitary transformations, and
so $f$ is not a self-test.

Now suppose (A) holds but (B) does not hold.  Then
there exist at least two inequivalent maxima $(\alpha_0, \ldots, \alpha_n)$
and $(\alpha'_0, \ldots, \alpha'_n)$.
The qubit strategies $T (\alpha_0, \ldots, \alpha_n)$
and $T( \alpha'_0, \ldots, \alpha'_n)$ are optimal
strategies whose measurement operators
are not compatible under local unitary transformation.
Therefore $f$ is not a self-test.  This completes the proof.
\end{proof}

The reader may note one consequence of
this proof: if a binary XOR game $f$ is a self-test, then it
is a self-test for the extended GHZ state $\frac{1}{\sqrt{2}}
\left( \left| 00 \ldots 0 \right> + \left| 11 \ldots 1 \right>
\right)$.  (Note also that, as a consequence
of \cite{McKague:2009}, we know that for every $n$
there is at least one binary XOR game
which self-tests the $n$-qubit GHZ state.)

\paragraph*{Robustness.}
Let us say that two qubit
strategies
$( \psi , \{ \{ N_j^{(0)} , 
N_j^{(1)} \} \}_j )$
and $( \gamma , \{
\{ S_j^{(0)} , S_j^{(1)} \}
\}_j )$
are \textit{$\delta$-close} if
\begin{eqnarray}
\left\| \psi - \gamma \right\| \leq \delta \hskip0.2in
\textnormal{and} \hskip0.2in
\left\| N_j^{(i)} - S_j^{(i)} \right\| \leq \delta
\end{eqnarray}
for all $j \in \{ 1, 2, \ldots, n \}$ and $i \in \{ 0, 1 \}$.
Let us say that a binary nonlocal XOR game
$f \colon \{ 0, 1 \}^n \to \mathbb{R}$ is
a \textit{second-order robust self-test}
if both condition (*) and the following
condition hold:
\begin{itemize}
\item[(**)] There exists a constant $C > 0$
such that any qubit
strategy whose score is within $\epsilon$
of the optimal score is $(C \sqrt{\epsilon})$-close to an optimal
qubit strategy.
\end{itemize}

In order to understand this condition, it useful to use Hessian matrices.
For any twice-differentiable function $F \colon \mathbb{R}^{m}
\to \mathbb{R}$ and any element $\mathbf{c} = (c_1 , \ldots , c_m )
\in \mathbb{R}^m$, let
\begin{eqnarray}
\Hess_\mathbf{c} ( F ) = \left[ \frac{\partial^2 F}{\partial x_i \partial x_j} (
\mathbf{c}) \right]_{i, j}.
\end{eqnarray}
When a $(2 \pi)$-periodic function $G$ is such that the Hessian matrices
at all its maxima are nonsingular, then it satisfies the following condition:
there is a constant $C$ such that for any $\mathbf{y} \in \mathbb{R}^m$ satisfying
$G ( \mathbf{y} ) \geq G_{max} - \epsilon$, there is a maximum $\mathbf{y}'$
for $G$ such that $\left\| \mathbf{y} - \mathbf{y}' \right\| \leq C \sqrt{\epsilon}$.
(See Lemma 1.1 in the supplementary information.)  

Consider the set $\mathbb{T} = \left\{ T ( \theta_0 , \theta_1, \ldots, \theta_n )
\mid \theta_i \in \mathbb{R} \right\}$.
The score achieved by $T ( \theta_0 , \theta_1 \ldots, \theta_n )$
at the game $f$ is given by $Z_f ( \theta_0, \theta_1, \ldots, \theta_n)$.
This motivates us to consider games $f$ that satisfy the following condition.
\begin{itemize}
\it \item[(C)] The maxima of $Z_f$ have nonzero Hessian matrices.
\end{itemize}
It is easy to see that $f$ satisfies this condition if and only if
condition (**) holds within the particular class of strategies $\mathbb{T}$.

The following proposition is proved in the supplementary information (see Proposition 6.1).
The proof is based on extending the reasoning above to arbitrary qubit strategies.

\begin{proposition}
\label{robustprop}
Let $f \colon \{ 0, 1 \}^n \to \mathbb{R}$
be a binary nonlocal XOR game.
Then, $f$ is a second-order robust self-test
if and only if it satisfies conditions (A), (B), and (C). \qed
\end{proposition}

\paragraph*{General quantum strategies.}

Now suppose that we consider quantum strategies of arbitrary finite dimension.
Whenever there are two Hermitian operators
$M^{(0)}, M^{(1)}$ on a single
finite-dimensional Hilbert space $\mathcal{Q}$, each having eigenvalues
in the set $\{ -1, 1 \}$, there exists a decomposition
$\mathcal{Q} = \bigoplus_{\ell = 1}^{m} \mathcal{Q}_{\ell}$
which is respected by both of the operators $M^{(0)}, M^{(1)}$, 
with $\dim \mathcal{Q}_{\ell} \leq 2$.  (See Lemma 3.3 in the supplementary
information.)
This allows us to reduce general quantum strategies to $n$-qubit
strategies.

The following generalization of Proposition~\ref{robustprop} is
proven in the supplementary information.  (See Proposition 4.2 
and Theorem 6.2.)

\begin{proposition}
Let $f \colon \{ 0, 1 \}^n \to \mathbb{R}$ be a binary nonlocal XOR game which satisfies conditions
(A), (B), and (C).  Then, there exists a constant $K > 0$ and
an $n$-qubit state $\chi \in \left( \mathbb{C}^2 \right)^{\otimes n}$ such that
the following holds: for any quantum strategy
\begin{eqnarray}
\Phi \in \mathcal{Q}_1 \otimes \ldots \otimes \mathcal{Q}_n \\
M_j^{(i)} \colon \mathcal{Q}_j \to \mathcal{Q}_j
\end{eqnarray}
achieving a score of $q_f - \epsilon$, there exist unitary embeddings $U_j \colon
\mathcal{Q}_j \to \mathbb{C}^2 \otimes \mathcal{Q}'_j$ and a vector
$\Gamma \in \mathcal{Q}'_1 \otimes \ldots \otimes \mathcal{Q}'_n$ such that
\begin{eqnarray}
\left\| \left( U_1 \otimes \cdots \otimes U_n \right)
\Phi - \chi \otimes \Gamma \right\| \leq K \sqrt{\epsilon}.  \qed
\end{eqnarray}
\end{proposition}

\paragraph*{Examples.}  It is easy now to show that the
function $Z_g$ corresponding
the CHSH game (\ref{chshfunc}) satisfies conditions (A)
and (B) with $(\alpha_0, \alpha_1, \alpha_2 ) = (- \frac{\pi}{4},
\frac{\pi}{2}, \frac{\pi}{2} )$.  The Hessian matrix
at this maximum is
\begin{eqnarray}
\left( - \frac{1}{\sqrt{2}} \right) \left[
\begin{array}{ccc} 4 & 2 & 2 \\
2 & 2 & 1 \\
2 & 1 & 2 \end{array} \right],
\end{eqnarray}
which is a nonsingular matrix.  Therefore, the CHSH game
is a second-order robust self-test.  

Let $d$ be the $3$-player GHZ game:
\begin{eqnarray*}
Z_d ( \theta_0, \theta_1 , \theta_2 , \theta_3 )  =
- \cos ( \theta_0)  
+ \cos ( \theta_0 + \theta_1  + \theta_2 ) \\
 + \cos ( \theta_0 + \theta_2 + \theta_3 ) 
 + \cos ( \theta_0 + \theta_1 + \theta_3 ).
\end{eqnarray*}
It easy to show that all maxima
of this function are equivalent to $(0 , \frac{\pi}{2} , \frac{\pi}{2} ,
\frac{\pi}{2}  )$,
and that the Hessian matrix at this maximum is
nonsingular.  Therefore (as was already known from \cite{mck:2010}), the GHZ game is also a self-test
that satisfies second-order robustness.

The recent paper \cite{amp:2012} by Ac{\'i}n \textit{et al.}
considers a family of
nonlocal games $\{ h_\alpha \colon \{ 0, 1 \}^2 \to \mathbb{R} \}_{\alpha > 1}$
defined by
\begin{eqnarray}
\begin{array}{cc}
h_\alpha (0, 0) = \alpha & h_\alpha ( 0, 1 ) = \alpha \\
h_\alpha (1, 0) = 1 & h_\alpha ( 1, 1 ) = -1.
\end{array}
\end{eqnarray}
The optimal score for this game is $2 \sqrt{\alpha^2 + 1}$.
The authors characterize the qubit-devices that achieve the
optimal score, and show that
these devices achieve more randomness 
than optimal devices for the standard CHSH inequality.
The games $h_\alpha$ may therefore be suitable
for randomness expansion.  For such an application, one needs to
also know how robust the tests are.

With the aid of the theory in \cite{amp:2012},
one can show that the function $Z_{h_\alpha} (
\theta_0, \theta_1, \theta_2 )$ has a unique maximum
up to equivalence, and the Hessian matrix
at this maximum is
\begin{eqnarray}
- (1 + \alpha^2)^{-1/2} \left[ \begin{array}{ccc}
2 \alpha^2 + 2 & \alpha^2 + 1 & 2 \\
\alpha^2 + 1 & \alpha^2 + 1 & 1 \\
2 & 1 & 2
\end{array} \right] 
\end{eqnarray}
which is nonsingular for any $\alpha > 1$.  Therefore, each
of the games in the family $\{ h_\alpha \}_{\alpha > 1}$
is a second-order robust self-test.
(We note also that the robustness coefficient $C$
from (**) can be determined
from the values of the function $Z_{h_\alpha} ( \theta_0,
\theta_1, \theta_2)$, as can be seen in the full proof of
Proposition~\ref{robustprop}.  This information
could be useful in choosing the most
appropriate value of $\alpha$.)

\paragraph{Open problems.}
A natural goal is to go beyond the class of XOR games
and identify the set of all nonlocal games that are robust self-testing.
A possible next step would be to consider games in
which the score is based on the XOR of a subset
of the outputs (as in the tests used \cite{mck:2010}), or a more general Boolean function.
A related general question is to determine if all pure entangled
states admit a robust self-test. Answers to those questions
will shed light on the power and limitations of classically
interacting with unknown quantum states and measurements,
and facilitate the development and the analysis of 
new quantum cryptographic 
protocols.

\paragraph{Acknowledgments.} The authors
thank Ryan Landay and Evan Noon for discussions
that helped with the development of this material.

\bibliographystyle{apsrev}
\bibliography{approx}

\end{document}


\title{Supplementary Information}

\maketitle

\section{Lemmas for multivariable functions}

In this section we prove some lemmas that will be used in later sections.
Each of these lemmas is concerned with the relationship between maxima 
and near-maxima of real-valued functions.

For any twice-differentiable function $F \colon \mathbb{R}^n \to \mathbb{R}$
and any $\mathbf{c} = (c_1 , \ldots, c_n ) \in \mathbb{R}^n$, let $\Hess_\mathbf{c}
F$ denote the Hessian matrix,
\begin{eqnarray*}
\Hess_\mathbf{c}
F & := &
\left[ \begin{array}{ccccc}
\frac{\partial^2 }{\partial z_1^2} F &
\frac{\partial}{\partial z_1} \frac{\partial}{ \partial z_2} F & 
\ldots & 
\frac{\partial}{\partial z_1} \frac{\partial}{ \partial z_n} F  \\
\frac{\partial}{\partial z_2} \frac{\partial}{ \partial z_1} F & 
\frac{\partial^2}{\partial z_2^2} F  &
\ldots & 
\frac{\partial}{\partial z_2} \frac{\partial}{ \partial z_n} F  \\
 \vdots & \vdots & \ddots & \vdots \\
\frac{\partial}{\partial z_n} \frac{\partial}{\partial z_1} F &
\frac{\partial}{\partial z_n} \frac{\partial}{\partial z_2} F &
\ldots & 
\frac{\partial^2 }{\partial z_n^2}  F
\end{array} \right] ( c_1, \ldots, c_n).
\end{eqnarray*}

\begin{lemma}
\label{2normlemma}
Let $G \colon \mathbb{R}^n \to \mathbb{R}$ be a twice-differentiable
function.  Let $\mathbf{x} = (x_1, \ldots, x_n ) \in \mathbb{R}^n$
be a local maximum of $G$ which is such that the Hessian matrix
$\Hess_{\mathbf{x}} ( G )$
is negative definite.  Then, there exist constants
$\delta_1, C_1 > 0$ such that for all $\mathbf{y} \in \mathbb{R}^n$
with $\left\| \mathbf{y} - \mathbf{x} \right\|
< \delta_1$,
\begin{eqnarray}
\left\| \mathbf{y} - \mathbf{x} \right\| \leq
C_1 \sqrt{ G ( \mathbf{x} ) -
G ( \mathbf{y} )}
\end{eqnarray}
\end{lemma}

\begin{proof}
Let $\lambda$ be the eigenvalue of $\Hess_{\mathbf{x}} ( G )$
which is closest to zero.
Define a function $F \colon \mathbb{R}^n \to \mathbb{R}$ by
\begin{eqnarray}
F ( \mathbf{y} ) & = & G ( \mathbf{y} ) + \left| \lambda/2
\right| \cdot \left\| \mathbf{y}
 - \mathbf{x} \right\|^2.
\end{eqnarray}
Then, the Hessian of $F$ at $\mathbf{x}$ is
\begin{eqnarray}
\Hess_{\mathbf{x}} ( G ) + \left| \lambda/2 \right| \mathbb{I}_n,
\end{eqnarray}
which is negative definite.  Therefore, $\mathbf{x}$ is a local maximum
of $F$.  So, there exists a constant $\delta_1$ such that
for all $\mathbf{y} \in \mathbb{R}^n$ with $\left\| \mathbf{y}
- \mathbf{x} \right\| \leq \delta_0$,
\begin{eqnarray}
G ( \mathbf{y} ) + \left| \lambda/2 \right| \left\| \mathbf{y}
- \mathbf{x} \right\|^2 \leq G ( \mathbf{x} ).
\end{eqnarray}
By an algebraic manipulation, the above inequality is equivalent
to
\begin{eqnarray}
\left\| \mathbf{y} - \mathbf{x} \right\| \leq
\left| \lambda/2 \right|^{-1} \sqrt{ G ( \mathbf{x} )
- G ( \mathbf{y} ) }.
\end{eqnarray}
This completes the proof.
\end{proof}

For any vector $\mathbf{y} = (y_1, \ldots, y_n ) \in \mathbb{R}^n$, let
$\left\| \mathbf{y} \right\|_\infty = \max_i \left| y_i \right|$.  Note that
\begin{eqnarray}
\left\| \mathbf{y} \right\|_\infty \leq \left\| \mathbf{y} \right\| \leq
\sqrt{n} \left\| \mathbf{y} \right\|_\infty.
\end{eqnarray}
The following modification of Lemma~\ref{2normlemma} follows immediately.
\begin{lemma}
\label{infnormlemma}
Let $G \colon \mathbb{R}^n \to \mathbb{R}$ be a twice-differentiable
function.  Let $\mathbf{x} = (x_1, \ldots, x_n ) \in \mathbb{R}^n$
be a local maximum of $G$ which is such that the Hessian matrix
$\Hess_{\mathbf{x}} ( G )$
is negative definite.  Then, there exist constants
$\delta_2, C_2 > 0$ such that for all $\mathbf{y} \in \mathbb{R}^n$
with $\left\| \mathbf{y} - \mathbf{x} \right\|_\infty
< \delta_2$,
\begin{eqnarray}
\left\| \mathbf{y} - \mathbf{x} \right\|_\infty \leq
C_2 \sqrt{ G ( \mathbf{x} ) -
G ( \mathbf{y} )}.  \qed
\end{eqnarray}
\end{lemma}

\begin{lemma}
\label{nearmaximalemma}
Let $G \colon \mathbb{R}^n \to \mathbb{R}$ be a twice-differentiable
function.  Suppose that there is a region in $U \in \mathbb{R}^n$ of
the form
\begin{eqnarray}
U = [a_1, b_1 ] \times [a_2 , b_2 ] \times \ldots \times
[a_n , b_n ]
\end{eqnarray}
such that $G$ has a single global maximum $\mathbf{z}$
in $U$.  Suppose
also that $\Hess_{\mathbf{z}} ( G )$ is negative definite.  Then,
there exists a constant $C_3 > 0$ such that the following holds:
for any $\mathbf{y} \in U$,
\begin{eqnarray}
\label{keyinequality}
\left\| \mathbf{y} - \mathbf{z} \right\|_\infty
\leq C_3 \sqrt{ G ( \mathbf{y} ) - G ( \mathbf{z} )
}.
\end{eqnarray}
\end{lemma}

\begin{proof}
Choose $C_2, \delta_2 > 0$ according to Lemma~\ref{infnormlemma}.
Let $r$ be the maximum value achieved by $G$ on the compact
set
\begin{eqnarray}
\left\{ \mathbf{y} \mid \mathbf{y} \in U, 
\left\| \mathbf{y} - \mathbf{z} \right\|_\infty \geq \delta_2 \right\}.
\end{eqnarray}
Note that this quantity is strictly less than $G ( \mathbf{z} )$.
Now, simply let
\begin{eqnarray}
C_3 = \max \left\{ C_2 , \frac{ \max \{ b_i - a_i \}_i } {\sqrt{G ( 
\mathbf{z})  - r} } \right\}.
\end{eqnarray}
Inequality (\ref{keyinequality}) is trivially satisfied when
$\left\| \mathbf{y} - \mathbf{x} \right\|_\infty \geq
\delta_2$, and by Lemma~\ref{infnormlemma} it is also satisfied
when $\left\| \mathbf{y} - \mathbf{x} \right\|_\infty <
\delta_2$.
\end{proof}

\begin{lemma}
\label{quadformlemma}
Let $H$ be an $m \times m$ real symmetric matrix, and let
\begin{eqnarray}
Q ( \mathbf{x} ) & = & \mathbf{x}^\top H \mathbf{x}
\end{eqnarray}
denote the corresponding real quadratic form.  Suppose
that $H$ has multiple eigenvalues, and let
$h_1 > h_2 > \ldots > h_r$ denote the eigenvalues
of the matrix $H$ taken in decreasing order.  
Then, for any unit vector $\mathbf{y} \in \mathbb{R}^m$,
there exists a unit vector $\mathbf{z} \in \mathbb{R}^m$
such that $Q ( \mathbf{z} ) = h_1$ and
\begin{eqnarray}
\left\| \mathbf{z} - \mathbf{y} \right\|
\leq \sqrt{\frac{2 \left( Q( \mathbf{z} )
- Q( \mathbf{y}) \right)}{h_1 - h_2}}.
\end{eqnarray}
\end{lemma}

\begin{proof}
Adjusting $H$ and $Q$ by an orthogonal linear transformation if necessary,
we may assume that
\begin{eqnarray}
Q ( \mathbf{x} ) = \sum_{\ell = 1}^m c_\ell x_i^2,
\end{eqnarray}
and that $h_1 = c_1 = c_2 = \ldots = c_b$ and $h_2 = c_{b+1}
\geq c_{b+2} \geq \ldots \geq c_m$,
for some $b \in \{ 2, 3, \ldots, m \}$.

Suppose that $\mathbf{y} \in \mathbb{R}^m$ is a unit vector
and that $Q ( \mathbf{y} ) = h_1 - \delta$. Then,
\begin{eqnarray}
h_1 - \delta & \leq & \sum_{\ell = 1}^m c_\ell y_\ell^2 \\
& \leq & h_1 \left( \sum_{\ell=1}^b y_\ell^2 \right) +
h_2 \left( \sum_{\ell=b+1}^m y_\ell^2 \right),
\end{eqnarray} 
from which it follows by an algebraic manipulation that
\begin{eqnarray}
y_{1}^2 + y_{2}^2 + \ldots + y_b^2 \geq 1 -\frac{\delta}{h_1 - h_2}.
\end{eqnarray}
Therefore, if we let
\begin{eqnarray}
\mathbf{z} = \left( y_1^2 + \ldots + y_b^2 \right)^{-1/2}
\left( y_1 , y_2, \ldots, y_b, 0, 0, \ldots, 0 \right)
\end{eqnarray}
we have
\begin{eqnarray}
\left\| \mathbf{z} - \mathbf{y} \right\|^2 & = & 2 - 2 \left< \mathbf{z} ,
\mathbf{y} \right> \\
& = & 2 - 2 \sqrt{ y_1^2 + \ldots + y_b^2 } \\
& \leq & 2 - 2 \left( y_1^2 + \ldots + y_b^2 \right) \\
& \leq & \frac{2 \delta}{h_1 - h_2}.
\end{eqnarray}
The desired result follows.
\end{proof}

\section{A robust self-testing result for qubit strategies}

\label{qubitstproofsection}

\begin{proposition}
\label{mainprop}
Let $f \colon \{ 0, 1 \}^n \to \mathbb{R}$
be a binary nonlocal XOR game which satisfies the following
conditions:
\begin{itemize}
\item[(1)] The function $Z_f$ has exactly two
maxima in the set $[-\pi , \pi]^{n+1}$,
and these maxima have the form
$(\alpha_0 , \ldots , \alpha_n)$
and $(- \alpha_0 , \ldots , - \alpha_n )$
with $0 < \alpha_i < \pi$
for all $i \in \{ 1, 2, \ldots, n \}$ and
$0 \leq \alpha_0 < \pi$.
\item[(2)] The Hessian matrices at each of
these maxima are nonsingular.
\end{itemize}
Then, $f$ is a second-order robust self-test.
\end{proposition}

We will prove Proposition~\ref{mainprop} by first addressing
some special classes of qubit strategies.
For any sequence of real numbers $(\theta_0 , \theta_1, \ldots ,
\theta_n )$, let $T ( \theta_0 , \ldots, \theta_n )$ denote
the $n$-qubit strategy given by the operators
\begin{eqnarray}
\left\{ \left(
\left[ \begin{array}{cc} 0 & 1 \\ 1 & 0
\end{array} \right] ,
\left[ \begin{array}{cc} 0
& e^{i \theta_j } \\ e^{- i \theta_j } & 0 \end{array}
\right]
\right) \mid j \in \{  1, 2, \ldots, n \} \right\}
\end{eqnarray}
and the unit vector
\begin{eqnarray}
\left( \frac{1}{\sqrt{2}} \right)
\left| 00 \ldots 0 \right> + \left( \frac{ e^{i \theta_0}
}{\sqrt{2}} \right) \left| 11 \ldots 1 \right>.
\end{eqnarray}

\begin{lemma}
\label{tlemma}
Let $f$ be a binary nonlocal XOR game which
satisfies conditions (1)--(2) from Proposition~\ref{mainprop}.
Let
\begin{eqnarray}
\mathbb{T} & = & \left\{ T \left( \theta_0 , \theta_1 , \ldots ,
\theta_n \right) \mid \theta_i \in [0, \pi] \textnormal{ for all }
i \in \{ 1, \ldots n \}, \textnormal{ and} \right.  \\
\nonumber & & \left. \theta_0 \in [ -\pi, \pi ] \right\}.
\end{eqnarray}
Then, there exists a constant $K_1 > 0$ such that
the following holds: any qubit strategy in $\mathbb{T}$
which achieves a score of $q_f - \epsilon$ must be
$(K_1 \sqrt{\epsilon})$-close to 
$T ( \alpha_0, \alpha_1 , \ldots , \alpha_n )$.
\end{lemma}

\begin{proof}
The score achieved by the strategy $T ( \theta_0 , \theta_1 ,
\ldots , \theta_n )$ is simply equal to
the quantity $Z_f ( \theta_0 , \theta_1 , \ldots, 
\theta_n )$.  The desired result follows easily by applying
Lemma~\ref{nearmaximalemma} to the function $Z_f$.
\end{proof}

Let $\mathbb{S}$ denote the set of all $n$-qubit strategies
$( \gamma , \{ \{ M_j^{(i)} \} \}_j )$ which are such that
the operators have the form
\begin{eqnarray}
M_j^{(0)} = \left[ \begin{array}{cc} 0 & 1 \\ 1 & 0 \end{array}
\right] \hskip0.5in M_j^{(1)} = \left[ \begin{array}{cc}
0 & e^{i \theta_j} \\ e^{-i \theta_j} & 0 \end{array} \right]
\end{eqnarray}
with $\theta_j \in [ 0 , \pi ]$, and the $n$-qubit state $\gamma$ has the form \begin{eqnarray}
\gamma = \sum_{i_1, \ldots, i_n \in \{ 0, 1 \}}
\gamma_{i_1 \cdots i_n} \left| i_1 \ldots i_n \right>
\end{eqnarray}
with $\gamma_{00\ldots 0} \geq 0$.

\begin{lemma}
\label{slemma}
Let $f$ be a binary nonlocal XOR game which 
satisfies conditions (1)--(2) from Proposition~\ref{mainprop}.
Then, there exists a constant $K_2$ such that
the following holds: any $n$-qubit strategy in $\mathbb{S}$
which achieves a score of $q_f - \epsilon$ must
be $( K_2 \sqrt{\epsilon } )$-close to some $n$-qubit strategy
in $\mathbb{T}$ which achieves an equal or higher score.
\end{lemma}

\begin{proof}
First we specify a value for the constant $K_2$.
Let $A$ denote the set of all $(n+1)$-tuples
$( \beta_0 , \beta_1 , \ldots , \beta_n )$ in
$[-\pi, \pi]^{n+1}$ which
are such that $0 < \beta_j < \pi$ for all 
$j \in \{ 1, 2, \ldots, n \}$.  Let
$A'$ denote the set of all $(n+1)$-tuples
$(\beta_0, \beta_1 , \ldots , \beta_n )$
in $[-\pi, \pi]^{n+1}$ which are such that
$- \pi < \beta_j < 0$ for all
$j \in \{ 1, 2, \ldots, n \}$.
Let $q'_f$ denote the maximum value
achieved by the function $Z_f$ on the
compact set $[- \pi , \pi ]^{n+1} \smallsetminus
(A \cup A' )$.  Note that this quantity is strictly
smaller than $q_f$.   Let
\begin{eqnarray}
K_2 = \frac{2}{\sqrt{ q_f - q'_f }}.
\end{eqnarray}

Let $( \gamma , \{ \{ M_j^{(i)} \} \}_j )$ be an $n$-qubit
strategy whose operators have the form
\begin{eqnarray}
M_j^{(0)} = \left[ \begin{array}{cc} 0 & 1 \\ 1 & 0 \end{array}
\right] \hskip0.5in M_j^{(1)} = \left[ \begin{array}{cc}
0 & e^{i \theta_j} \\ e^{-i \theta_j} & 0 \end{array} \right]
\end{eqnarray}
with $\theta_j \in [ 0 , \pi ]$, and whose $n$-qubit state $\gamma$ has the form \begin{eqnarray}
\gamma = \sum_{i_1, \ldots, i_n \in \{ 0, 1 \}}
\gamma_{i_1 \cdots i_n} \left| i_1 \ldots i_n \right>
\end{eqnarray}
with $\gamma_{00\ldots 0} \geq 0$.
Suppose that this strategy achieves a score of $q_f - \epsilon$.
We wish to show that $( \gamma , \{ \{ M_j^{i)} \} \}_j )$ is 
$K_2 \sqrt{\epsilon}$ close to a strategy in $\mathbb{T}$ which achieves
an equal or higher score.  If $\epsilon \geq q_f - q'_f$, this is trivial
to prove.  So, let us assume that $\epsilon < q_f - q'_f$.

Let us construct an expression for the score
for the strategy in terms of the function $Z_f$.
Let
\begin{eqnarray}
\mathbf{M} = \sum_{i_1, \ldots, i_n \in \{ 0, 1 \}}
f ( i_1, \ldots, i_n ) M_1^{(i_1)} \otimes 
\ldots M_n^{(i_n)}.
\end{eqnarray}
Let us write $\mathbf{0}$ and $\mathbf{1}$ for the sequences
$(0, 0, \ldots, 0)$ and $(1, 1, \ldots, 1 ) \in \{ 0, 1 \}^n$, and
for any binary sequence $\mathbf{i} \in \{ 0, 1 \}^n$,
let us write $\mathbf{1} - \mathbf{i}$ for the 
sequence $( 1 - i_1 , \ldots, 1 - i_n)$. 
We have
\begin{eqnarray}
q_f - \epsilon & = & \left< \gamma \right| \mathbf{M} \left| \gamma \right> \\
& = & \sum_{i_1, \ldots, i_n \in \{ 0, 1 \}} \gamma_\mathbf{i}
\overline{\gamma_{\mathbf{1} - \mathbf{i}}}
P_f ( e^{i (-1)^{i_1} \theta_1} , \ldots , e^{i (-1)^{i_n} \theta_n} )
\end{eqnarray}
Combining conjugate terms, we have
\begin{eqnarray}
q_f - \epsilon
& = & \sum_{\substack{\mathbf{i} \in \{ 0, 1 \}^n \\
i_1 = 0}} 2 \textnormal{Re} \left[ \gamma_\mathbf{i}
\overline{\gamma_{\mathbf{1} - \mathbf{i}} } P_f
( e^{i (-1)^{i_1}  \theta_1  } , \ldots, e^{i (-1)^{i_n} \theta_n} ) \right]
\end{eqnarray}
Let $\gamma_\mathbf{i} = r_{\mathbf{i}} e^{it_\mathbf{i}}$,
with $r_\mathbf{i} \geq 0$ and $t_{\mathbf{i}} \in [ -\pi ,
\pi ]$.  We may assume that $t_\mathbf{0} = 0$. Then,
\begin{eqnarray}
q_f - \epsilon & = & \sum_{\substack{\mathbf{i} \in \{ 0, 1 \}^n \\
i_1 = 0}} 2 r_\mathbf{i} r_{\mathbf{1} - \mathbf{i}} \textnormal{Re} \left[ e^{i t_{\mathbf{i}}} e^{-i t_\mathbf{1 - i}} P_f
( e^{i (-1)^{i_1} \theta_1} , \ldots, e^{i (-1)^{i_n} \theta_n} ) \right]
\end{eqnarray}
and therefore
\begin{eqnarray}
\label{scoreformula}
q_f - \epsilon & = & \sum_{\substack{\mathbf{i} \in \{ 0, 1 \}^n \\
i_1 = 0}} 2 r_\mathbf{i} r_{\mathbf{1} - \mathbf{i}} Z_f
( t_\mathbf{i} - t_{\mathbf{1} - \mathbf{i}} , (-1)^{i_1} \theta_1
, \ldots, (-1)^{i_n} \theta_n )
\end{eqnarray}

Note that
\begin{eqnarray}
\sum_{\substack{\mathbf{i} \in \{ 0, 1 \}^n \\ i_1 = 0}} 2 r_\mathbf{i}
r_{\mathbf{1} - \mathbf{i}} \leq \sum_{\substack{\mathbf{i} \in \{ 0, 1 \}^n
\\ i_1 = 0}} \left( r_\mathbf{i}^2 + r_{\mathbf{1} - \mathbf{i}}^2 \right) = 1.
\end{eqnarray}
Therefore the positive quantity $( q_f - \epsilon )$ satisfies
\begin{eqnarray}
q_f - \epsilon \leq  \max_{\substack{\mathbf{i} \in \{ 0, 1 \}^n \\ i_1 = 1}}
Z_f ( t_{\mathbf{i}} - t_{\mathbf{1} - \mathbf{i}} , (-1)^{i_1} \theta_1 , \ldots ,
(-1)^{i_n} \theta_n ).
\end{eqnarray}
All elements of the set
\begin{eqnarray}
\{ Z_f \left( t_\mathbf{i} - t_{\mathbf{1} - \mathbf{i}} ,
(-1)^{i_1} \theta_1 , \ldots , (-1)^{i_n} \theta_n \right) \mid
\mathbf{i} \in \{ 0, 1 \}^n , i_0 = 0 \},
\end{eqnarray}
aside from $Z_f \left( t_\mathbf{0} - t_\mathbf{1} ,
\theta_1 , \ldots, \theta_n \right)$, are bounded above by $q'_f < q_f - \epsilon$.
Therefore, the quantity $Z_f \left( t_\mathbf{0} - t_\mathbf{1} ,
\theta_1 , \ldots, \theta_n \right) ( = 
Z_f \left(  t_\mathbf{1} ,
\theta_1 , \ldots, \theta_n \right) )$ must be at least
$q_f - \epsilon$.
We conclude that the strategy
\begin{eqnarray}
\label{tstrategy}
T ( t_\mathbf{1} , \theta_1 , \ldots, \theta_n ) \in \mathbb{T}
\end{eqnarray}
achieves a score at least as high as that of the of the original
strategy $( \gamma , \{ \{ M_j^{(i)} \} \}_j )$.

To see that strategy (\ref{tstrategy}) is $(K_2 \sqrt{\epsilon})$-close to
$( \gamma , \{ \{ M_j^{(i)} \} \}_j )$, note that equality
(\ref{scoreformula}) implies
\begin{eqnarray}
q_f - \epsilon & \leq & 2 r_\mathbf{0} r_{\mathbf{1} - \mathbf{1}} (
q_f ) + \sum_{\substack{ \mathbf{i} \in \{ 0, 1 \}^n \\ i_1 = 0 \\ \mathbf{i} \neq \mathbf{0}}}
2 r_\mathbf{i} r_{\mathbf{1} - \mathbf{i}} ( q'_f ).
\end{eqnarray}
The expression on the right above is a quadratic form on the
unit vector $(r_\mathbf{i})$.  By Lemma~\ref{quadformlemma},
\begin{eqnarray}
\sqrt{\left( r_\mathbf{0} - \frac{1}{\sqrt{2}} \right)^2
+ \left( r_\mathbf{1} - \frac{1}{\sqrt{2}} \right)^2 
+ \sum_{\mathbf{i} \neq \mathbf{0}, \mathbf{1}} r_\mathbf{i}^2}
\leq \sqrt{ \frac{2 \epsilon}{q_f - q'_f}}.
\end{eqnarray}
The expression on the left side of this inequality
is equal to the distance between the vector $\gamma$
and the vector
\begin{eqnarray}
\frac{1}{\sqrt{2}} \left| 00 \ldots 0 \right> +
\frac{e^{it_1}}{\sqrt{2}} \left| 11 \ldots 1 \right>.
\end{eqnarray}
Therefore,
strategy (\ref{tstrategy}) is $\sqrt{2 \epsilon / (q_f - q'_f )}$-close
to $( \gamma , \{ \{ M_j^{(i)} \} \}_j )$.  Since
$K_2 > \sqrt{2 / (q_f - q'_f) }$, this completes the proof.
\end{proof}

\begin{lemma}
\label{slemma2}
Let $f$ be a binary nonlocal XOR game which satisfies conditions (1)--(2) from
Proposition~\ref{mainprop}.  Then, there exists a constant $K_3$ such that the following
holds:
any qubit strategy in $\mathbb{S}$ which achieves a score of $q_f - \epsilon$ must
be $(K_3 \sqrt{\epsilon})$-close to the (optimal) strategy
$T ( \alpha_0 , \alpha_1 , \ldots, \alpha_n )$.
\end{lemma}

\begin{proof}
This follows easily from Lemmas~\ref{tlemma} and \ref{slemma}, with
$K_3 = K_1 + K_2$.
\end{proof}

\begin{proof}[Proof of Proposition~\ref{mainprop}]
For any $n$-qubit strategy $( \phi , \{ \{ N_j^{(0)} , N_j^{(1)} \} \}_j )$,
there are unitary transformations $U_j \colon \mathbb{C}^2 \to \mathbb{C}^2$
for $j = 1, 2, \ldots, n$ such that the strategy given by the vector 
\begin{eqnarray}
\left( U_1 \otimes U_2 \otimes \ldots \otimes U_n \right)
\phi
\end{eqnarray}
and the operators
\begin{eqnarray}
U_j N_j^{(i)} U_j^{-1}
\end{eqnarray}
is in the class $\mathbb{S}$.  Proposition~\ref{mainprop} therefore 
follows from Lemma~\ref{slemma2}.
\end{proof}

\section{Jordan's lemma}

\begin{lemma}
\label{dim2decomplemma}
Let $W = \mathbb{C}^n$, with $n \geq 1$, and let $M_1$ and $M_2$ be projection operators
on $W$.  Then, there exists a nonzero subspace $Y \subseteq W$ with $\dim Y \leq 2$
such that $Y$ is stabilized by both $M_1$ and $M_2$.
\end{lemma}

\begin{proof}
Let $V_1, V_2 \subseteq W$ be the images of the operators $M_1$
and $M_2$, respectively.
Consider the quantity
\begin{eqnarray}
\max_{\substack{v \in V_1, v' \in V_2 \\ | v | = | v' | = 1}} \left| \left< v, v' \right> \right|.
\end{eqnarray}
Let $v_1 \in V_1$ and $v_2 \in V_2$ be vectors which achieve this 
maximum.  We carry out the proof in 3 cases.

\vskip0.2in 
\textbf{Case 1:} $\left| \left< v_1, v_2 \right> \right| = 0$.

Let $Y$ be the span of $v_1$ in $W$.  Then $M_1$ clearly 
stabilizes $Y$.  The space $V_2$ is orthogonal to $v_1$, 
so the projection operator $M_2$ maps $Y$ to $\{ 0 \}$.
Therefore $Y$ is stabilized by both $M_1$ and $M_2$.

\vskip0.2in
\textbf{Case 2:} $\left| \left< v_1, v_2 \right> \right| = 1$.

Let $Y$ be the span of $v_1$ in $W$.  The space $Y$ is contained
in both $V_1$ and $V_2$, and therefore $Y$ is stabilized by
both $M_1$ and $M_2$.

\vskip0.2in
\textbf{Case 3:} $0 < \left| \left< v_1, v_2 \right> \right| < 1$.

Let $Y$ be the span of $\{ v_1, v_2 \}$.  We have assumed that the vector
$v_2$ is a vector which achieves
the maximum value for the inner product $\left| \left< v_1, v \right> \right|$ 
for all $v \in V_2$ with $\left| v \right| = 1$.  Therefore, $v_2$ must be
a scalar multiple of the projection $M_2 v_1$.  Thus, $M_2$ stabilizes
$Y$.   A similar argument shows that $M_1$ 
stabilizes $Y$ as well.
\end{proof}

\begin{lemma}
\label{dim2decomplemma2}
Let $W = \mathbb{C}^n$, with $n \geq 1$, and let $Z_1$ and $Z_2$ be Hermitian operators on
 $W$ whose eigenvalues are contained in the set $\{ -1, 1 \}$.
Then, there exists a nonzero subspace $Y \subseteq W$ with $\dim Y \leq 2$
such that $Y$ is stabilized by both $Z_1$ and $Z_2$.
\end{lemma}

\begin{proof}
This follows immediately by applying Lemma~\ref{dim2decomplemma} to
the projection operators $(Z_1 + \mathbb{I})/2$ and $(Z_2 + \mathbb{I})/2$.
\end{proof}

\begin{lemma}
\label{fulldecomplemma}
Let $W = \mathbb{C}^n$, with $n \geq 1$, and let $Z_1$ and $Z_2$ be Hermitian
operators on $W$ whose eigenvalues are contained in the set $\{ -1 , 1 \}$.
Then, there exists an orthogonal decomposition of $W$ into
subspaces $W_1, \ldots, W_r$ which is respected by $Z_1$ and
$Z_2$ such that $\dim W_i \leq 2$ for all $i \in \{ 1, \ldots, r \}$.
\end{lemma}

\begin{proof}
This follows easily from Lemma~\ref{dim2decomplemma2} by induction.
\end{proof}

The next lemma gives a canonical form for any pair of Hermitian
operators that have eigenvalues in the set $\{ -1, 1 \}$.

\begin{lemma}
\label{opcanonlemma}
Let $W = \mathbb{C}^n$ for some $n > 0$, and let $X_1$ and $X_2$ be Hermitian
operators on $W$ whose eigenvalues are contained in the set $\{ -1 , 1 \}$.  Then,
there exists a unitary embedding $U \colon W \to \mathbb{C}^{2m}$ for some $m > 0$ and Hermitian
operators $X'_1, X'_2$ satisfying $U^* X'_i U = X_i$ such that $X'_1, X'_2$ have the form
\begin{eqnarray}
\label{canon1}
X'_1 = \left[ \begin{array}{cccccccc}  0 & 1 & &&&&& \\
1 & 0 & &&&&& \\
& & 0 & 1 & &&& \\
& & 1 & 0 & & & & \\
& & & & \ddots & & & \\
& & & & & \ddots & & \\
& & & & & & 0 & 1 \\
& & & &  & & 1 & 0 
\end{array} \right],
\end{eqnarray}
\begin{eqnarray}
\label{canon2}
X'_2 = \left[ \begin{array}{cccccccc}  0 & e^{i \theta_1} & &&&&& \\
e^{-i \theta_1} & 0 & &&&&& \\
& & 0 & e^{i \theta_2} & &&& \\
& & e^{-i \theta_2} & 0 & & & & \\
& & & & \ddots & & & \\
& & & & & \ddots & & \\
& & & & & & 0 & e^{i \theta_m} \\
& & & &  & & e^{-i \theta_m} & 0 
\end{array} \right],
\end{eqnarray}
with $\theta_\ell \in [0, \pi]$ for all $\ell \in \{ 1, 2, \ldots , m \}$.
\end{lemma}

\begin{proof}
\textbf{Case 1:} $n = 1$.  In this case $X_1$ and $X_2$ are $1 \times 1$
matrices, which we denote by $[x_1]$ and $[x_2]$.  Let
\begin{eqnarray}
X'_1 & = & \left[ \begin{array}{cc} 0 & 1 \\ 1 & 0 \end{array} \right].
\end{eqnarray}
If $x_1 = 1$, then let $U \colon W \to \mathbb{C}^2$ be a map
whose image is the positive eigenspace of $U$; otherwise, let
$U$ be a map whose image is the negative eigenspace of $X'_1$.

If $x_2 = x_1$, then let $X'_2 = X'_1$; otherwise, 
let $X'_2 = -X'_1$.  The desired properties hold.

\vskip0.2in

\textbf{Case 2:} $n = 2$.  If either of $X_1$ or $X_2$ is a scalar
matrix, then we can find an orthogonal decomposition of
$W$ into $1$-dimensional subspaces that is respected by 
both operators, and we thus reduce to case 1.

If both $X_1$ and $X_2$ are nonscalar matrices, then each operator
has a nontrivial $(+1)$-eigenspace and a nontrivial $(-1)$-eigenspace.
We can find an orthonormal basis for $W$ which puts
$X_1$ and $X_2$ in the form of (\ref{canon1}) and (\ref{canon2}) above.

\vskip0.2in

\textbf{Case 3:} $n > 2$.  By Lemma~\ref{fulldecomplemma},
we may find an orthonormal bases for $W$ under which
the matrix expressions for $X_1$ and $X_2$ decompose
into $1 \times 1$ and $2 \times 2$ diagonal blocks.  The
desired result now follows from cases 1 and 2.
\end{proof}

\section{A canonical form for quantum strategies}

Let
\begin{eqnarray}
\label{astrategy1}
\left( \alpha , \left\{ \left\{ S_j^{(0)} , S_j^{(1)} \right\} \right\}_j \right), \\
\nonumber \alpha  \in \mathcal{A}_1 \otimes \mathcal{A}_2 \otimes \ldots \otimes \mathcal{A}_n.
\end{eqnarray}
and
\begin{eqnarray}
\label{astrategy2}
\left( \beta , \left\{ \left\{ T_j^{(0)} , T_j^{(1)} \right\} \right\}_j \right), \\
\nonumber \beta \in \mathcal{B}_1 \otimes \mathcal{B}_2 \otimes \ldots \otimes \mathcal{B}_n.
\end{eqnarray}
be quantum strategies.  Let us say that a \textit{unitary embedding} from strategy
(\ref{astrategy1}) to strategy (\ref{astrategy2}) is a collection of
unitary embeddings 
\begin{eqnarray}
\left\{ U_j \colon A_j \hookrightarrow B_j  \right\}_{j \in \{ 1, 2, \ldots, n \}}
\end{eqnarray}
such that 
\begin{eqnarray}
\left( U_1 \otimes \ldots \otimes U_n \right) \alpha = \beta
\end{eqnarray}
and
\begin{eqnarray}
U_j^* T^{(i)}_j U_j = S^{(i)}_j
\end{eqnarray}
for all $i \in \{ 0, 1 \}$, $j \in \{ 1, 2, \ldots , n \}$.

\begin{definition}
A quantum strategy
\begin{eqnarray}
\left( \beta , \left\{ \left\{ T_j^{(0)} , T_j^{(1)} \right\} \right\}_j \right), \\
\nonumber \beta  \in \mathcal{B}_1 \otimes \mathcal{B}_2 \otimes \ldots \otimes \mathcal{B}_n.
\end{eqnarray}
is in \textit{canonical form} if the following three properties hold for each $k \in \{ 1, \ldots, n \}$:
\begin{enumerate}
\item $\mathcal{B}_k = \mathbb{C}^2 \otimes \mathcal{W}_k$, where $\mathcal{W}_k$
is a finite-dimensional Hilbert space with a fixed orthonormal basis $\{ w_{k1}, w_{k2}, \ldots, w_{km_k} \}$.

\item The operator $T_k^{(0)}$ has the form
\begin{eqnarray}
T_k^{(0)} & = & \sum_{\ell = 1}^{m_k} \left[ \begin{array}{cc} 0 & 1 \\ 1 & 0 \end{array} \right] \otimes
\left| w_{k \ell} \right> \left< w_{k \ell} \right| 
\end{eqnarray}

\item The operator $T_k^{(1)}$ has the form
\begin{eqnarray}
T_k^{(1)} & = & \sum_{\ell = 1}^{m_k} \left[ \begin{array}{cc} 0 & e^{i \theta_{k \ell}} \\
e^{-i \theta_{k \ell}}  & 0 \end{array} \right] \otimes
\left| w_{k \ell} \right> \left< w_{k \ell} \right| 
\end{eqnarray}
with $\theta_{k \ell} \in [0, \pi ]$ for all $\ell \in \{ 1 ,  \ldots , m_k \}$.
\end{enumerate}
\end{definition}

\begin{proposition}
Any quantum strategy has a unitary embedding into a quantum strategy in canonical form.
\end{proposition}

\begin{proof}
This follows easily from Lemma~\ref{opcanonlemma}.
\end{proof}

\section{A robust self-testing result for quantum systems of arbitrary dimension}

\begin{theorem}
\label{mainthm}
Let $f$ be a nonlocal game which satisfies the following two conditions:
\begin{enumerate}
\item The function $Z_f \colon [ -\pi , \pi ]^{n+1} \to \mathbb{R}$
has exactly two global maxima, and the maxima have the form
$( \alpha_0, \alpha_1 , \ldots
, \alpha_n )$ and $( -\alpha_0 , - \alpha_1 , \ldots , - \alpha_n )$,
with $0 < \alpha_j < \pi$ for all $j \in \{ 1, 2, \ldots, n \}$, and
$0 \leq \alpha_0 < \pi$.

\item The Hessian matrices of $Z_f$ at each of these maxima are nonsingular.
\end{enumerate}
Then, there exists an $n$-qubit state $g \in \left( \mathbb{C}^2 \right)^{\otimes n}$
and a constant $K > 0$ such that the following holds:
\begin{itemize}
\item For any quantum strategy in canonical form,
\begin{eqnarray*}
\Lambda \in (\mathbb{C}^2 \otimes \mathcal{W}_1 )  \otimes \ldots
\otimes (\mathbb{C}^2 \otimes \mathcal{W}_n ) \\
\left\{ M_j^{(i)} \colon \mathbb{C}^2 \otimes \mathcal{W}_j
\to \mathbb{C}^2 \otimes \mathcal{W}_j \right\}_{i,j},
\end{eqnarray*}
which achieves score $q_f - \epsilon$,
there exists a unit vector $\gamma \in \mathcal{W}_1 \otimes \ldots
\otimes \mathcal{W}_n$ such that 
\begin{eqnarray*}
\left\| \Lambda - g \otimes \gamma \right\|
\leq C \sqrt{\epsilon}.
\end{eqnarray*}
\end{itemize}
\end{theorem}

\begin{proof}
Let
\begin{eqnarray}
g & = & \frac{1}{\sqrt{2}} \left( \left| 00 \ldots 0 \right>
+ \frac{ P_f ( \alpha_1 , \ldots , \alpha_n )}{ | P_f ( \alpha_1 , \ldots , \alpha_n ) |}
\left| 11 \ldots 1 \right> \right).
\end{eqnarray}

Let
\begin{eqnarray}
\label{Lambdastrategy}
\Lambda \in (\mathbb{C}^2 \otimes \mathcal{W}_1 )  \otimes \ldots
\otimes (\mathbb{C}^2 \otimes \mathcal{W}_n ) \\
\nonumber \left\{ M_j^{(i)} \colon \mathbb{C}^2 \otimes \mathcal{W}_j
\to \mathbb{C}^2 \otimes \mathcal{W}_j \right\}_{i,j},
\end{eqnarray}
be a quantum strategy in canonical form which achieves score $q_f - \epsilon$.
We may write the operators $\{ M_j^{(1)} \}_{j=1}^n$ as
\begin{eqnarray}
M_j^{(1)} = \sum_{\ell = 1}^{m_j} \left[ \begin{array}{cc}
0 & e^{i \theta_{j \ell}} \\ e^{-i \theta_{j \ell}} & 0 \end{array} \right] 
\otimes \left| w_{j \ell} \right> \left< w_{j \ell} \right|,
\end{eqnarray}
where for each $j$, $\{ w_{j1} , \ldots, w_{jm_j} \}$ is
an orthonormal basis for $\mathcal{W}_j$, and
each element $\theta_{j \ell}$ is a real number from
the interval $[0, \pi ]$.
Write the state $\Lambda$ as
\begin{eqnarray}
\Lambda = \sum_{\substack{\ell_1, \ldots, \ell_n \\
0 < \ell_k \leq m_k}} \left( p_{\ell_1 \ldots \ell_n} \right)
\lambda_{\ell_1 \ldots \ell_n} \otimes w_{1 \ell_1} \otimes \ldots \otimes w_{n \ell_n}.
\end{eqnarray}
where each $\lambda_{\ell_1 \ldots \ell_n}$ denotes a unit
vector in $( \mathbb{C}^2 )^{\otimes n}$ and 
$\{ p_{\ell_1 \ldots \ell_n} \}$ is a set of complex
numbers satisfying $\sum \left| p_{\ell_1 \cdots \ell_n} \right|^2 = 1$.  By adjusting the vectors $\{ \lambda_{\ell_1, \ldots, \ell_n} \}$
and the quantities $\{ p_{j \ell } \}$ by scalar multiplication if necessary, we may assume
that
\begin{eqnarray}
\left< \lambda_{\ell_1 \ell_2 \ldots \ell_n} \mid 00 \ldots 0 \right> \geq 0
\end{eqnarray}
for each of the vectors $\lambda_{\ell_1 \ell_2 \ldots \ell_n}$.

For any $n$-tuple $(\ell_1, \ldots, \ell_n )$ with $1 \leq \ell_k \leq
m_k$, there is a qubit strategy determined by the vector
\begin{eqnarray}
\label{componentvector}
\lambda_{\ell_1
\ell_2 \ldots \ell_n }
\end{eqnarray}
and the operators
\begin{eqnarray}
\label{componentoperators}
\left\{ \left\{ \left[ \begin{array}{cc} 0 & 1 \\ 1 & 0 \end{array} \right] ,
\left[ \begin{array}{cc} 0 & e^{i \theta_{j \ell_j}} \\ e^{- i \theta_{j \ell j}} &
0 \end{array} \right] \right\} \right\}_{j=1}^n.
\end{eqnarray}
Let $s_{\ell_1 \ldots \ell_n}$ denote the score achieved by this strategy.  The score
for strategy (\ref{Lambdastrategy}) is simply a weighted average of the
scores $\{ s_{\ell_1 \ldots \ell_n} \}$:
\begin{eqnarray}
q_f - \epsilon = \sum_{1 \leq \ell_j \leq m_j} \left| p_{\ell_1 \ldots \ell_n} \right|^2
s_{\ell_1 \ldots \ell_n}.
\end{eqnarray}

Each strategy (\ref{componentvector})--(\ref{componentoperators})
is a member of the class $\mathbb{S}$ defined in section~\ref{qubitstproofsection}.
By Lemma~\ref{slemma2}, there is a constant $K$ such that
\begin{eqnarray}
\left\| \lambda_{\ell_1 \ldots \ell_n} - g \right\| \leq K \sqrt{q_f - s_{\ell_1 \ldots \ell_n}}
\end{eqnarray}
for all $n$-tuples $(l_j)$ with $1 \leq \ell_j \leq m_j$.
Let
\begin{eqnarray}
\gamma & = & \sum_{1 \leq \ell_j \leq m_j} p_{\ell_1 \ldots \ell_n}
w_{1\ell_1 } \otimes w_{2 \ell_2} \otimes \ldots \otimes w_{n \ell_n}.
\end{eqnarray}
Then,
\begin{eqnarray}
\left\| \Lambda - g \otimes \gamma \right\|^2 & = & 
\sum_{1 \leq \ell_j \leq n } \left| p_{\ell_1 \ldots \ell_n} \right|^2 \left\| \lambda_{\ell_1 \ldots
\ell_n } - g \right\|^2 \\
& \leq & \sum_{1 \leq \ell_j \leq n } \left| p_{\ell_1 \ldots \ell_n} \right|^2 K^2
\left( q_f - s_{\ell_1 \ldots \ell_n} \right) \\
& = & K^2 \left[ q_f - \sum_{1 \leq \ell_j \leq m_j} \left| p_{\ell_1 \ldots \ell_n} \right|^2
s_{\ell_1 \ldots \ell_n} \right] \\
& = & K^2 \left[ q_f - (q_f - \epsilon ) \right] \\
& = & K^2 \epsilon.
\end{eqnarray}
This completes the proof.
\end{proof}

\section{Generalizations}

For any
binary nonlocal XOR game $f \colon \{ 0, 1 \}^n \to \mathbb{R}$,
and any binary sequence $b_0, b_1, \ldots, b_n \in \{ 0, 1 \}^n$,
we can define a binary nonlocal XOR game $g \colon 
\{ 0, 1 \}^n \to \mathbb{R}$ defined by
\begin{eqnarray}
\label{modifiedgame}
g ( i_1, i_2 , \ldots , i_n ) = (-1)^{b_0} f ( b_1 \oplus i_1 ,
b_2 \oplus i_2 , \ldots, b_n \oplus i_n ).
\end{eqnarray}
Note that for any quantum strategy $( \phi , \{ \{ M_j^{(0)} , M_j^{(1)} \} \}_j )$
for $f$, one can construct a quantum strategy for $g$ using
the same state $\phi$ which achieves the same
score.  (Simply exchange the two measurements $M_j^{(0)}$
and $M_j^{(1)}$ whenever $b_j = 1$, and,
if $b_0 = 1$, then also negate both $M_1^{(0)}$ and $M_1^{(1)}$.)
Clearly, $f$ is a second-order robust self-test if and only if
$g$ is a second-order robust self-test.

Note that $Z_f$ and $Z_g$ are related as follows:
\begin{eqnarray*}
Z_g ( \theta_0, \ldots, \theta_n ) =
(-1)^{b_0} Z_f \left( \theta_0 + \sum_i b_i \theta_i , (-1)^{b_1} \theta_1 ,
(-1)^{b_2} \theta_2 , \ldots, (-1)^{b_n} \theta_n \right).
\end{eqnarray*}

\begin{proposition}
\label{generalmainprop}
Let $f \colon \{ 0, 1 \}^n \to \mathbb{R}$
be a binary nonlocal XOR game.
Then, $f$ is a second-order robust self-test
if and only if it satisfies the following three conditions.
\begin{itemize}
\item[(A)] There is a maximum $(\alpha_0, \alpha_1, \ldots , \alpha_n)$
for $Z_f$ such that none of $\alpha_1, \alpha_2,  \ldots, \alpha_n$ is a multiple of $\pi$.

\item[(B)] Every other maximum of $Z_f$ is congruent modulo $2 \pi$ 
to either $(\alpha_0, \ldots, \alpha_n)$ or $(-\alpha_0, \ldots, - \alpha_n)$.

\item[(C)] The maxima of $Z_f$ have nonzero Hessian matrices.
\end{itemize}
\end{proposition}

\begin{proof}
Suppose that $f \colon \{ 0, 1 \} \to \mathbb{R}$ is such that conditions
(A)--(C) hold.  Then, replacing $f$ if necessary with a game of the
form (\ref{modifiedgame}), we may assume that $f$ has a maximum 
$(\beta_0, \ldots, \beta_n )$ which satisfies $\beta_0 \in 
[ 0, \pi )$ and $\beta_1 , \ldots , \beta_n \in (0, \pi )$.  Then, by
Proposition~\ref{mainprop}, $f$ is a second-order robust self-test.

Now suppose conversely that $f \colon \{ 0, 1 \} \to \mathbb{R}$
is a second-order robust self-test.  Then conditions (A) and (B)
follow from Proposition 2 in the main text, so we need
only prove condition (C).  Recall from
section~\ref{qubitstproofsection} that 
for any $(n+1)$-tuple $( \theta_0, \ldots , \theta_n )$
there is an associated qubit strategy
$T ( \theta_0 , \ldots , \theta_n )$
which achieves a score of $Z_f ( \theta_0, \ldots, \theta_n )$.
Second-order robustness within the class $\mathbb{T}$
implies the following: there exists a constant $D$
such that for any $(\theta_0 , \ldots, \theta_n )$
satisfying $Z_f ( \theta_0, \ldots, \theta_n ) \geq q_f 
- \epsilon$, there is  a maximum $(\theta'_0, \ldots,
\theta'_n )$ satisfying $\sum_k \left| \theta_k - \theta'_k \right|
\leq D \sqrt{\epsilon}$.  As a consequence,
all the maxima of $Z_f$ must have nonsingular
Hessian matrices.  This completes the proof.
\end{proof}

The next theorem is a generalization of Theorem~\ref{mainthm}
which follows easily using construction (\ref{modifiedgame}).

\begin{theorem}
Let $f$ be a nonlocal game which satisfies the conditions (A), (B), and (C) from 
Proposition~\ref{generalmainprop}.
Then, there exists an $n$-qubit state $g \in \left( \mathbb{C}^2 \right)^{\otimes n}$
and a constant $K > 0$ such that the following holds:
\begin{itemize}
\item For any quantum strategy in canonical form,
\begin{eqnarray*}
\Lambda \in (\mathbb{C}^2 \otimes \mathcal{W}_1 )  \otimes \ldots
\otimes (\mathbb{C}^2 \otimes \mathcal{W}_n ) \\
\left\{ M_j^{(i)} \colon \mathbb{C}^2 \otimes \mathcal{W}_j
\to \mathbb{C}^2 \otimes \mathcal{W}_j \right\}_{i,j},
\end{eqnarray*}
which achieves score $q_f - \epsilon$,
there exists a unit vector $\gamma \in \mathcal{W}_1 \otimes \ldots
\otimes \mathcal{W}_n$ such that 
\begin{eqnarray*}
\left\| \Lambda - g \otimes \gamma \right\|
\leq K \sqrt{\epsilon}. \qed
\end{eqnarray*}
\end{itemize}
\end{theorem}


\title{Randomness Expansion from the Greenberger-Horne-Zeilinger Paradox\footnote{This
appendix was written during June - August 2011.  The text given here is unmodified except
for a few minor edits.}}

\author{Carl A.~Miller \and Yaoyun Shi}

\date{}

\maketitle

\appendix

\section{Randomness Expansion Devices}

In this section, we formalize  the notion of a multi-part \textit{randomness expansion} (RE) device.

A \textbf{$3$-part randomness expansion (RE) device} $D$ is a device consisting of $3$ components,
$D_1$, $D_2$, $D_3$, satisfying all of the following conditions:
\begin{enumerate}
\item Each component $D_j$ accepts a single bit $i_j \in \{ 0 , 1 \}$ as input and returns a single bit $o_j
\in \{ 0, 1 \}$
as output.  

\item Each component $D_j$ contains a quantum system, $Q_j$.  When it receives 
its input bit $i_j$, it performs an orthogonal binary measurement on $Q_j$.
The output bit $o_j$ is the result of this measurement.

\item The choice of measurement performed by $D_j$ is determined solely
by its input bit $i_j$.  In particular, it does not depend on the inputs or outputs of any of the other
components.  (The components do not communicate with one another.)
\end{enumerate}

We will say that the input bits $(i_j)$ and the output bits $(o_j)$
\textbf{pass the GHZ test} if the following Boolean relations
are satisfied:
\begin{eqnarray}
\label{ghztest1}
i_1 \oplus i_2 \oplus i_3 = 0, \\
\label{ghztest2}
o_1 \oplus o_2 \oplus o_3 \oplus (i_1 \vee i_2 \vee i_3) = 1.
\end{eqnarray}

In what follows, we will often be choosing the inputs to device $D$ according to
a random probability distribution.  In this case, we denote the inputs by 
the random variables $I_1$, $I_2$, and $I_3$, and the outputs by the random variables
$O_1$, $O_2$, and $O_3$.  Let $I = (I_1, I_2 , I_3)$, and $O = (O_1 , O_2 , O_3)$.
\[
\xymatrix{
I_1 \ar[d] & I_2 \ar[d] & I_3 \ar[d] \\
\framebox{ \framebox[0.5in]{$Q_1$} }_{D_1} \ar[d] &
\framebox{ \framebox[0.5in]{$Q_2$} }_{D_2} \ar[d] &
\framebox{ \framebox[0.5in]{$Q_3$} }_{D_3} \ar[d] \\
O_1 & O_2 & O_3
}
\]
Let $Q$ denote the tripartite quantum system formed by $Q_1$, $Q_2$, and $Q_3$.
As a convention, if $Z$ is a quantum system, then we use
the corresponding symbol $\mathcal{Z}$ to denote the complex
Euclidean state-space of $Z$.    Thus $\mathcal{Q}_1$, $\mathcal{Q}_2$,
and $\mathcal{Q}_3$ denote the state spaces of $Q_1$, $Q_2$, and $Q_3$.

Note that the state of the system $Q$ may be mixed.  Sometimes we will consider an
additional quantum system that is entangled with $Q$:
\[
\xymatrix{
I_1 \ar[d] & I_2 \ar[d] & I_3 \ar[d] \\
\framebox{ \framebox[0.5in]{$Q_1$} }_{D_1} \ar[d] &
\framebox{ \framebox[0.5in]{$Q_2$} }_{D_2} \ar[d] &
\framebox{ \framebox[0.5in]{$Q_3$} }_{D_3} \ar[d] &
\framebox[0.5in]{$E$}  \\
O_1 & O_2 & O_3
}
\]
We will use the expressions $\Gamma^{(pre)}_*$ and $\Gamma^{(post)}_*$ to
denote the pre- and post-measurements states of the systems above.
Thus,
\begin{eqnarray}
\Gamma^{(post)}_{IQE}
\end{eqnarray}
denotes the post-measurement
state of the systems
$I$, $Q$, and $E$, taken together.  
We will drop the
``(pre)'' and ``(post)'' superscripts when they are not necessary.

$\Gamma$ is always presumed to denote a density operator.
If the state of a collection of quantum systems happens to be pure,
then we may also also use the symbol $v$ to denote the vector representing
the pure state.  (Thus for example, if the
expression $v^{(pre)}_Q$ appears, it should be
interpreted to mean a vector in $\mathcal{Q}$ representing
the pure pre-measurement state of $\mathcal{Q}$.)

\subsection{A canonical form}

Let $D$ be a $3$-part RE device (as defined above).  For any $j \in \{ 1, 2, 3 \}$,
and any input bit $i \in \{ 0 , 1 \}$, let
\begin{eqnarray}
\label{examplemeasurement}
\{ P_j^{(i , o )} \}_{o \in \{ 0 , 1 \} }
\end{eqnarray}
be the orthogonal measurement the the component $D_j$ performs on 
its quantum system $Q_j$ on input $i$.  This is a set consisting
of two complementary orthogonal projection operators on
the space $\mathcal{Q}_j$.

It is convenient to express (\ref{examplemeasurement}) as a single
operator.  Let
\begin{eqnarray}
T_j^{(i)} = P_j^{(i, 0)} - P_j^{(i, 1)}.
\end{eqnarray}
This is a Hermitian operator on $\mathcal{Q}_j$ whose square
is equal to $\mathbb{I}$.  In general, any Hermitian operator
$X$ satisfying $X^2 = \mathbb{I}$ yields a binary orthogonal
measurement ($\{ \frac{ \mathbb{I} + X }{2} , \frac{\mathbb{I} - X}{2} \}$).
For this reason, we will often refer to such an operator $X$ simply as a ``measurement.''

Let $D'$ be another $3$-part RE device,
\[
\xymatrix{
\framebox{ \framebox[0.5in]{$Q'_1$} }_{D'_1} &
\framebox{ \framebox[0.5in]{$Q'_2$} }_{D'_2} &
\framebox{ \framebox[0.5in]{$Q'_3$} }_{D'_3}  \\
}
\]
whose measurements are denoted by ${T'}_j^{(i)}$.  Let us
say that a \textbf{unitary embedding} of $D$ into $D'$ is
a collection of unitary embeddings
\begin{eqnarray*}
\Psi_1 \colon \mathcal{Q}_1 \hookrightarrow \mathcal{Q}'_1  \\
\Psi_2 \colon \mathcal{Q}_2 \hookrightarrow \mathcal{Q}'_2  \\
\Psi_3 \colon \mathcal{Q}_3 \hookrightarrow \mathcal{Q}'_3 
\end{eqnarray*}
which map the pre-measurement state of $Q = (Q_1 , Q_2, Q_3 )$ to
the pre-measurement state of $Q' = (Q'_1 , Q'_2 , Q'_3)$, and which 
are such that
\begin{eqnarray}
T_j^{(i)} = \Psi_j^\dagger \left(  {T'}_j^{(i)} \right) \Psi_j
\end{eqnarray}
It is easy to see that if such a unitary embedding exists,
then the device $D'$ has the same output-statistics as
$D$.

If $E$ is an additional quantum system, and the pre-measurement
states of $D$ and $D'$ are given as entangled states
\begin{eqnarray*}
\Gamma^{(pre)}_{QE} \textnormal{ and } {\Gamma'}^{(pre)}_{Q'E},
\end{eqnarray*}
then we say that the unitary embedding $(\Psi_i )$ \textbf{respects 
entanglement with $E$} if the superoperator
\begin{eqnarray}
\left( \Psi_1 \otimes \Psi_2 \otimes \Psi_3 \otimes \mathbb{I}_{\mathcal{E}} 
\right) ( \cdot ) \left( \Psi_1 \otimes \Psi_2 \otimes \Psi_3 \otimes \mathbb{I}_{\mathcal{E}} 
\right)^\dagger
\end{eqnarray}
carries
$\Gamma^{(pre)}_{QE}$ to ${\Gamma'}^{(pre)}_{Q'E}$.

\begin{proposition}
\label{canonicalform}
Let $D$ be a $3$-part RE device which is entangled with an additional
quantum system $E$.  Then there exists a unitary embedding of $D$
into another $3$-part RE device $D'$, respecting
entanglement with $E$, such that the following conditions
hold.
\begin{enumerate}
\item  The state-spaces
of $D'$ have the form
\begin{eqnarray}
\mathcal{Q}'_1 &  = & \mathcal{A} \otimes \mathcal{R}_1 \\
\mathcal{Q}'_2 & = & \mathcal{B} \otimes \mathcal{R}_2 \\
\mathcal{Q}'_3 & = & \mathcal{C} \otimes \mathcal{R}_3.
\end{eqnarray}
where $\mathcal{A}$, $\mathcal{B}$, and $\mathcal{C}$ are complex Euclidean
spaces, and $\mathcal{R}_1$, $\mathcal{R}_2$, and $\mathcal{R}_3$
are copies of $\mathbb{C}^{\{ 0 , 1 \}}$.

\item The measurements operators for $D'_j$ on input $i = 0$ are given,
respectively, by
\begin{eqnarray}
\sum_k \left| a_k \right> \left< a_k \right| \otimes \left[
\begin{array}{cc}
0 & 1 \\ 1 & 0 \end{array} \right], \\
\sum_\ell \left| b_\ell \right> \left< b_\ell \right| \otimes \left[
\begin{array}{cc}
0 & 1 \\ 1 & 0
\end{array} \right],\\
\sum_m \left| c_m \right> \left< c_m \right| \otimes \left[
\begin{array}{cc}
0 & 1 \\ 1 & 0
\end{array}
\right],
\end{eqnarray}
where $\{ a_k \}$, $\{ b_\ell \}$, and $\{ c_m \}$ denote orthogonal
bases for $\mathcal{A}$, $\mathcal{B}$, and $\mathcal{C}$.\footnote{The
$2 \times 2$ matrices here are written with
respect to the computational basis $\{ \left| 0 \right> , \left| 1 \right> \}$
of $\mathbb{C}^{\{ 0 , 1 \}}$.}

\item The measurements operators for $D'_j$ on input
$i = 1$ are given by
\begin{eqnarray}
\sum_k \left| a_k \right> \left< a_k \right| \otimes \left[
\begin{array}{cc}
0 & \lambda_k \\ \overline{\lambda_k} & 0 \end{array} \right], \\
\sum_\ell \left| b_\ell \right> \left< b_\ell \right| \otimes \left[
\begin{array}{cc}
0 & \gamma_\ell \\ \overline{\gamma_\ell} & 0
\end{array} \right],\\
\sum_m \left| c_m \right> \left< c_m \right| \otimes \left[
\begin{array}{cc}
0 & \phi_m \\ \overline{\phi_m} & 0
\end{array}
\right],
\end{eqnarray}
where $\lambda_k$, $\gamma_\ell$, and $\phi_m$
denote complex numbers that satisfy
\begin{eqnarray}
\left| \lambda_k \right| = \left| \gamma_\ell \right|
= \left| \phi_m \right| = 1 , \\
\Im ( \lambda_k ) \geq 0 , \\
\Im ( \gamma_\ell ) \geq 0 , \\
\Im ( \phi_m ) \geq 0.
\end{eqnarray}
\end{enumerate}
\end{proposition}

\begin{proof}
It is easy to show that, given any two Hermitian operators
$X$ and $Y$ on $\mathbb{C}^2$ such that $X^2 = Y^2 = 
\mathbb{I}$, there is a unitary transformation $U \colon
\mathbb{C}^2 \to \mathbb{C}^2$ such that
\begin{eqnarray}
U X U^\dagger & = & \left[ \begin{array}{cc}
0 & 1 \\
1 & 0 \end{array} \right] \\
U Y U^\dagger & = & \left[ \begin{array}{cc}
0 & \zeta \\
\overline{ \zeta } & 0 \end{array} \right],
\end{eqnarray}
where $\zeta$ satisfies $\left| \zeta \right| = 1$
and $\Im ( \zeta ) \geq 0$.  From this fact,
it is easy to prove the proposition for
the case where the spaces $\mathcal{Q}_j$ all
have dimension $\leq 2$.  
The general case then follows via Jordan's lemma.\footnote{See Lemma 1 from  \url{http://www.cs.tau.ac.il/~odedr/teaching/quantum_fall_2005/ln/qma.pdf}.}
\end{proof}

\section{The GHZ Paradox: $2 \times 2 \times 2$ case}

\label{222ghzsection}

We are interested in those RE devices which pass the
GHZ test (see (\ref{ghztest2})) with high probability.
Because of Proposition~\ref{canonicalform}, it is useful to
focus on RE devices that are in a particular form.

Let $J$ be a $3$-part RE device.
\[
\xymatrix{
\framebox{ \framebox[0.5in]{$R_1$} }_{J_1} &
\framebox{ \framebox[0.5in]{$R_2$} }_{J_2} &
\framebox{ \framebox[0.5in]{$R_3$} }_{J_3} \\
}
\]
Suppose that $\mathcal{R}_j = \mathbb{C}^{\{ 0 , 1 \}}$ for each $j$, and suppose
that the measurement
operators $\{ S^{(i)}_j \}$ for $J$ are given by
\begin{eqnarray} 
S^{(0)}_1 = S^{(0)}_2 = S^{(0)}_3 = \left[
\begin{array}{cc} 0 & 1 \\ 1 & 0 \end{array}
\right] \\
S^{(1)}_1  = \left[
\begin{array}{cc} 0 & \lambda \\ \overline{\lambda} & 0 \end{array}
\right] \\
S^{(1)}_2  = \left[
\begin{array}{cc} 0 & \gamma \\ \overline{\gamma} & 0 \end{array}
\right] \\
S^{(1)}_3  = \left[
\begin{array}{cc} 0 & \phi \\ \overline{\phi} & 0 \end{array}
\right]
\end{eqnarray}
with $\left| \lambda \right| = \left| \gamma \right|
= \left| \phi \right| = 1$.  Suppose that the pre-measurement
state of $(R_1 , R_2 , R_3 )$ is a pure state given by a 
vector $\alpha \in \mathcal{R}_1 \otimes \mathcal{R}_2
\otimes \mathcal{R}_3$.

Consider that a $3$-bit string $I_1 I_2 I_3$ is chosen uniformly
at random from the set
\begin{eqnarray}
\label{simpleinputset}
\left\{ 000 , 011 , 101 , 110 \right\},
\end{eqnarray}
and given to device $J$.  Then $J$ passes the GHZ test if
and only if the output string $O_1O_2O_3$ satisfies the relation
\begin{eqnarray}
O_1 \oplus O_2 \oplus O_3 \oplus \left( I_1 \vee I_2 \vee I_3 \right)
 = 1.
\end{eqnarray}
We can calculate the probability that this formula will be satisfied
using the operators $S^{(i)}_j$.  For example, the
probability that the formula
will be satisfied when the input string is $110$ is:
\begin{eqnarray}
\frac{1}{2} & + & \frac{1}{2} \left< \left( S^{(1)}_1 \otimes
S^{(1)}_2 \otimes S^{(0)}_3 \right) \alpha , \alpha \right>
\end{eqnarray}
The probability that this formula will be satisfied on an input
randomly chosen from (\ref{simpleinputset}) is 
\begin{eqnarray}
\label{opexpforGHZ}
\frac{1}{2} & - & \frac{1}{8} \left< \left( S^{(0)}_1 \otimes
S^{(0)}_2 \otimes S^{(0)}_3 \right) \alpha , \alpha \right> \\
\nonumber
& + & \frac{1}{8} \left< \left( S^{(0)}_1 \otimes
S^{(1)}_2 \otimes S^{(1)}_3 \right) \alpha , \alpha \right> \\
\nonumber
& + & \frac{1}{8} \left< \left( S^{(1)}_1 \otimes
S^{(0)}_2 \otimes S^{(1)}_3 \right) \alpha , \alpha \right> \\
\nonumber
& + & \frac{1}{8} \left< \left( S^{(1)}_1 \otimes
S^{(1)}_2 \otimes S^{(0)}_3 \right) \alpha , \alpha \right>
\end{eqnarray}
Let
\begin{eqnarray}
\alpha = \sum_{k, l, m \in \{ 0 , 1 \} }
c_{klm} \left| klm \right>.
\end{eqnarray}
Writing (\ref{opexpforGHZ}) in terms of
$\lambda$, $\gamma$, $\phi$, and $\{
c_{jkl} \}$, and rearranging terms, we find the following expression:
\begin{eqnarray*}
\mathbf{P} \left( \textnormal{$J$ passes the
GHZ test}\right)  \hskip0.05in =  \hskip0.05in 
\frac{1}{2}  & + &
\textnormal{Re} \left[ \overline{c_{111}} \cdot c_{000} \left( \frac{ -1 + \lambda \gamma +
\gamma \phi + \phi \lambda }{4} \right) \right] \\
\nonumber
& + &
\textnormal{Re} \left[ \overline{c_{110}} \cdot c_{001} \left( \frac{ -1 + \lambda \gamma +
\gamma \overline{\phi} + \overline{\phi } \lambda }{4} \right) \right] \\
\nonumber
& + &
\textnormal{Re} \left[ \overline{c_{101}} \cdot c_{010} \left( \frac{ -1 + \lambda \overline{ \gamma } +
\overline{ \gamma } \phi + \phi \lambda }{4} \right) \right] \\
\nonumber
& + &
\textnormal{Re} \left[ \overline{c_{100}} \cdot c_{011} \left( \frac{ -1 + \lambda \overline{ \gamma } +
\overline{ \gamma } \overline{ \phi } + \overline{ \phi } \lambda }{4} \right) \right].
\end{eqnarray*}
We will make frequent use of this expression in the subsections that follow.

\subsection{A Few Lemmas}

The following lemmas will be useful in subsequent sections.

\begin{lemma}
\label{threevariablelemma1}
Let $a$, $b$, and $c$ be complex numbers such that $\left| a \right| = \left| b \right|
= \left| c \right| = 1$ and $\Im ( a ) \geq 0$, $\Im ( b ) \geq 0$, and $\Im ( c ) \geq 0$.  Then,
\begin{eqnarray}
\left| \frac{ -1 + ab + bc + ca }{4} \right| 
\leq \sqrt{ 1 - \frac{ \left| a - i \right|^2 }{4}  }.
\end{eqnarray}
\end{lemma}

\begin{proof}
Consider, for all complex numbers $b'$ satisfying $| b' | = 1$, the value
\begin{eqnarray}
\left| -1 + ab' + b' c + ca \right| = \left| \left( -1 + ca \right)
+ b' \left( a + c \right) \right|.
\end{eqnarray}
This value will be maximized when the
angle of $\left( -1 + ca \right)$ in the complex plane agrees
with the angle of $b' ( a + c )$.  This occurs precisely when $b' = i$.  Therefore,
\begin{eqnarray}
\left| -1 + ab + bc + ca \right| \leq \left| -1 + ai + ic + ca \right|.
\end{eqnarray}
By similar reasoning,
\begin{eqnarray}
\left| -1 + ai + ic + ca \right| \leq \left| -1 + ai + i^2 + ia \right|.
\end{eqnarray}
Therefore,
\begin{eqnarray}
\left| -1 + ab + bc + ca \right| \leq \left| -2 + 2ai \right| =  \left| 2a + 2i \right| 
\end{eqnarray}
Note that $\left| 2a + 2i \right|^2 + \left| 2a - 2i \right|^2 = 16$.  Therefore,
\begin{eqnarray}
\left| -1 + ab + bc + ca \right|^2 & \leq & \left| 2 a + 2 i \right|^2
\end{eqnarray}
implies
\begin{eqnarray}
\left| -1 + ab + bc + ca \right|^2 & \leq & 16 - \left| 2a - 2i \right|^2.
\end{eqnarray}
Dividing by $16$ and taking square roots of both sides of the equation
yields the result.
\end{proof}

\begin{lemma}
\label{threevariablelemma2}
Let $a$, $b$, and $c$ be complex numbers such that $\left| a \right| = \left| b \right|
= \left| c \right| = 1$.  Suppose that $\textnormal{Im} (a ) \geq 0$ and $\Im ( b) \geq 0$,  but $\textnormal{Im} (c ) \leq 0$.  Then,
\begin{eqnarray}
\left| \frac{ -1 + ab + bc + ca}{4} \right| \leq \frac{\sqrt{2}}{2}.
\end{eqnarray}
\end{lemma}

\begin{proof}
Consider the complex numbers
\begin{eqnarray}
(-1 + ab) \hskip0.2in \textnormal{and} \hskip0.2in (a + b ).
\end{eqnarray}
In the complex plane, $(-1 + ab)$ lies at an angle of $+ \pi/2$ (in the counterclockwise
direction) from $(a + b)$.  Since $\Im ( c ) \leq 0$, the angle between the product
$c(a + b)$ and $(-1 + ab)$ must be an obtuse angle or a right angle.  Therefore,
\begin{eqnarray}
\left| (-1 + ab ) + c (a + b ) \right|^2 \leq \left| -1 + ab \right|^2
+ \left| a + b \right|^2 \leq 8.
\end{eqnarray}
The lemma follows.
\end{proof}

\begin{lemma}
\label{infinitynormlemma}
Suppose that $A$, $B$, $A'$ and $B'$ are Hermitian operators on $\mathbb{C}^n$ which
satisfy $\left| \left| \cdot \right| \right|_\infty \leq 1$.
Then,
\begin{eqnarray}
\left| \left| A \otimes B  - A' \otimes B' \right| \right|_\infty
& \leq & \left| \left| A - A' \right| \right|_\infty
+ \left| \left| B - B' \right| \right|_\infty.
\end{eqnarray}
\end{lemma}

\begin{proof}
We have the following:
\begin{eqnarray}
\left| \left| A \otimes B - A' \otimes B' \right| \right|_\infty
& = & \left| \left| (A - A' ) \otimes B + A' \otimes (B - B' ) \right| \right|_\infty \\
& \leq & \left| \left| A - A' \right| \right|_\infty \left| \left| B \right| \right|_\infty
+ \left| \left| A' \right| \right|_\infty \left| \left| B - B' \right| \right|_\infty \\
& \leq & \left| \left| A - A' \right| \right|_\infty + \left| \left| B - B' \right| \right|_\infty.
\end{eqnarray}
\end{proof}

\subsection{Characterizing Approximate GHZ Devices}

\label{characterizationsubsection}

There exists $3$-part RE devices which satisfy the GHZ test perfectly. Consider
the $3$-qubit device whose pre-measurement state is the pure state
\begin{eqnarray}
\label{ghzstate}
g = \frac{1}{\sqrt{2}} \left( \left| 000 \right> -
\left| 111 \right> \right) \in \mathbb{C}^{\{ 0 , 1 \}} \otimes
\mathbb{C}^{\{ 0 , 1 \}} \otimes \mathbb{C}^{\{ 0 , 1 \}} 
\end{eqnarray}
and which is the such that each component uses the measurement
operator
\begin{eqnarray}
\sigma_x = \left[ \begin{array}{cc} 0 & 1 \\
1 & 0 \end{array} \right]
\end{eqnarray}
on input $i = 0$, and the measurement operator
\begin{eqnarray}
\sigma_y = \left[ \begin{array}{cc} 0 & i \\
-i & 0 \end{array} \right]
\end{eqnarray}
on input $i = 1$.  This device passes the GHZ test with probability $1$.
Let us refer to this device as \textbf{the ideal GHZ device}.

We will use the above example for comparison.  In 
the next two propositions, we see that if a $2 \times 2 \times 2$
RE device passes the GHZ test with
high probability, then it is similar to an ideal GHZ device.

\begin{proposition}
\label{msrmtapproxprop}
Let $J$ be a $3$-part RE device whose state spaces 
$\mathcal{R}_j$ are copies of the qubit-space $\mathbb{C}^{
\{ 0 , 1 \} }$, and whose pre-measurement state
is a pure state
\begin{eqnarray}
\alpha = \sum_{k, l , m \in \{ 0 , 1 \} } c_{klm} \left| klm \right>.
\end{eqnarray}
Suppose that the measurement operators for $J$ are
\begin{eqnarray}
\begin{array}{ccc}
S_1^{(0)} =  \left[ \begin{array}{cc} 0 & 1 \\ 1 & 0 \end{array} \right]  & 
S_2^{(0)} = \left[ \begin{array}{cc} 0 & 1 \\ 1 & 0 \end{array} \right] & 
S_3^{(0)} = \left[ \begin{array}{cc} 0 & 1 \\ 1 & 0 \end{array} \right] \\ \\
S_1^{(1)} =  \left[ \begin{array}{cc} 0 & \lambda \\ \overline{\lambda} & 0 \end{array} \right]  & 
S_2^{(1)} = \left[ \begin{array}{cc} 0 & \gamma \\ \overline{\gamma} & 0 \end{array} \right] & 
S_3^{(1)} = \left[ \begin{array}{cc} 0 & \phi \\ \overline{\phi} & 0 \end{array} \right] \\ 
\end{array}
\end{eqnarray}
where
\begin{eqnarray}
\left| \lambda \right| = \left| \gamma \right|
= \left| \phi \right| = 1 \\
\Im ( \lambda ) \geq 0 \\
\Im ( \gamma ) \geq 0 \\
\Im ( \phi ) \geq 0.
\end{eqnarray}
Suppose that
\begin{eqnarray}
\mathbf{P} \left( \textnormal{$J$ passes the GHZ test} \right) = 1 - \epsilon.
\end{eqnarray}
Then,
\begin{eqnarray}
\label{lambdaineq}
\left| \lambda - i \right|^2 \leq 16 \epsilon \\
\label{gammaineq}
\left| \gamma - i \right|^2 \leq 16 \epsilon \\
\label{phiineq}
\left| \phi - i \right|^2 \leq 16 \epsilon.
\end{eqnarray}
\end{proposition}

\begin{proof}
From the discussion at the beginning of this section, we know that
the probability that $J$ passes the GHZ test is given by
\begin{eqnarray}
\label{expforpassprob}
\frac{1}{2} & + &
\textnormal{Re} \left[ \overline{c_{111}} \cdot c_{000} \left( \frac{ -1 + \lambda \gamma +
\gamma \phi + \phi \lambda }{4} \right) \right] \\
\nonumber
& + &
\textnormal{Re} \left[ \overline{c_{110}} \cdot c_{001} \left( \frac{ -1 + \lambda \gamma +
\gamma \overline{\phi} + \overline{\phi } \lambda }{4} \right) \right] \\
\nonumber
& + &
\textnormal{Re} \left[ \overline{c_{101}} \cdot c_{010} \left( \frac{ -1 + \lambda \overline{ \gamma } +
\overline{ \gamma } \phi + \phi \lambda }{4} \right) \right] \\
\nonumber
& + &
\textnormal{Re} \left[ \overline{c_{100}} \cdot c_{011} \left( \frac{ -1 + \lambda \overline{ \gamma } +
\overline{ \gamma } \overline{ \phi } + \overline{ \phi } \lambda }{4} \right) \right]
\end{eqnarray}
Therefore, the quantity $(1 - \epsilon)$ satisfies
\begin{eqnarray}
1 - \epsilon \hskip0.05in \leq \hskip0.05in \frac{1}{2} & + &
\left| c_{111} \right| \left|  c_{000} \right|  \left| \frac{ -1 + \lambda \gamma +
\gamma \phi + \phi \lambda }{4} \right| \\
\nonumber
& + &
\left|  c_{110} \right| \left| c_{001} \right| \left| \frac{ -1 + \lambda \gamma +
\gamma \overline{\phi} + \overline{\phi } \lambda }{4} \right| \\
\nonumber
& + &
\left| c_{101} \right| \left|  c_{010} \right| \left| \frac{ -1 + \lambda \overline{ \gamma } +
\overline{ \gamma } \phi + \phi \lambda }{4} \right| \\
\nonumber
& + &
\left| c_{100} \right| \left| c_{011} \right| \left| \frac{ -1 + \lambda \overline{ \gamma } +
\overline{ \gamma } \overline{ \phi } + \overline{ \phi } \lambda }{4} \right|
\end{eqnarray}
which implies
\begin{eqnarray}
\label{absvalueinequality}
1 - \epsilon \hskip0.05in \leq \hskip0.05in \frac{1}{2} & + & \frac{1}{2} \left[ 
\left( \left| c_{111}  \right|^2 +  \left|  c_{000} \right|^2 \right)  \left| \frac{ -1 + \lambda \gamma +
\gamma \phi + \phi \lambda }{4} \right| \right. \\
\nonumber
& + &
\left( \left| c_{110}  \right|^2 +  \left|  c_{001} \right|^2 \right)  \left| \frac{ -1 + \lambda \gamma +
\gamma \overline{\phi} + \overline{\phi} \lambda }{4} \right| \\
\nonumber
& + &
\left( \left| c_{101}  \right|^2 +  \left|  c_{010} \right|^2 \right)  \left| \frac{ -1 + \lambda \overline{ \gamma } +
\overline{ \gamma } \phi + \phi \lambda }{4} \right|  \\
\nonumber
& + &
\left. \left( \left| c_{100}  \right|^2 +  \left|  c_{011} \right|^2 \right)  \left| \frac{ -1 + \lambda \overline{ \gamma  } +
\overline{ \gamma } \overline{ \phi } + \overline{ \phi}  \lambda }{4} \right| \right],
\end{eqnarray}
or equivalently,
\begin{eqnarray}
\label{absvalueinequality2}
1 - 2 \epsilon & \leq &
\left( \left| c_{111}  \right|^2 +  \left|  c_{000} \right|^2 \right)  \left| \frac{ -1 + \lambda \gamma +
\gamma \phi + \phi \lambda }{4} \right|  \\
\nonumber
& + &
\left( \left| c_{110}  \right|^2 +  \left|  c_{001} \right|^2 \right)  \left| \frac{ -1 + \lambda \gamma +
\gamma \overline{\phi} + \overline{\phi} \lambda }{4} \right| \\
\nonumber
& + &
\left( \left| c_{101}  \right|^2 +  \left|  c_{010} \right|^2 \right)  \left| \frac{ -1 + \lambda \overline{ \gamma } +
\overline{ \gamma } \phi + \phi \lambda }{4} \right|  \\
\nonumber
& + &
\left( \left| c_{100}  \right|^2 +  \left|  c_{011} \right|^2 \right)  \left| \frac{ -1 + \lambda \overline{ \gamma  } +
\overline{ \gamma } \overline{ \phi } + \overline{ \phi}  \lambda }{4} \right|.
\end{eqnarray}

Consider the following four quantities from inequality (\ref{absvalueinequality2}):
\begin{eqnarray*}
\left| \frac{ -1 + \lambda \gamma +
\gamma \phi + \phi \lambda }{4} \right| , 
\left| \frac{ -1 + \lambda \gamma +
\gamma \overline{\phi} + \overline{\phi} \lambda }{4} \right| , 
\left| \frac{ -1 + \lambda \overline{ \gamma } +
\overline{ \gamma } \phi + \phi \lambda }{4} \right| , 
\left| \frac{ -1 + \lambda \overline{ \gamma  } +
\overline{ \gamma } \overline{ \phi } + \overline{ \phi}  \lambda }{4} \right|.
\end{eqnarray*}
Note that the coefficients of these quantities in (\ref{absvalueinequality2})
sum to $1$.
By Lemma~\ref{threevariablelemma2}, the
second, third, and fourth quantities are each less than or equal to 
$\sqrt{2}/2$.  If the first of the quantities were also less than or
equal to $\sqrt{2} / 2$, then we have $1 - 2 \epsilon \leq \frac{\sqrt{2}}{2}$,
and therefore $\epsilon \geq \frac{1}{2} - \frac{\sqrt{2}}{4} > \frac{1}{8}$,
and inequalities (\ref{lambdaineq}) - (\ref{phiineq}) would be trivially satisfied.  
So, let us assume that the first of these quantities is greater than $\sqrt{2}/2$, which
makes it the largest of the four.
Then,
\begin{eqnarray}
1 - 2 \epsilon & \leq &  \left| \frac{ -1 + \lambda \gamma + \gamma \phi
+ \phi \lambda}{4} \right|.
\end{eqnarray}
By Lemma~\ref{threevariablelemma1},
\begin{eqnarray}
1 - 2 \epsilon & \leq & \sqrt{ 1 - \frac{ \left| \lambda - i \right|^2}{4} }.
\end{eqnarray}
Therefore,
\begin{eqnarray}
\left( 1 - 2 \epsilon \right)^2 & \leq &  1 - \frac{ \left| \lambda - i \right|^2}{4} \\
1 - 4 \epsilon & \leq &  1 - \frac{ \left| \lambda - i \right|^2}{4} \\
16 \epsilon & \geq & \left| \lambda - i \right|^2
\end{eqnarray}
Similar reasoning shows that $\left| \gamma - i \right|^2 \leq 16 \epsilon$
and $\left| \phi - i \right|^2 \leq 16 \epsilon$.
\end{proof}

\begin{proposition}
\label{stateapproxprop}
Let $J$ be the device from Proposition~\ref{msrmtapproxprop}.  
Let 
\begin{eqnarray}
g = \frac{1}{\sqrt{2}} \left( \left| 000 \right> -
\left| 111 \right> \right)
\end{eqnarray}
Then, 
\begin{eqnarray}
\label{prestateinequality}
\left| \left< \alpha , g \right> \right| \geq 1 - 7 \epsilon.
\end{eqnarray}
\end{proposition}

\begin{proof}
As before, we know that the probability
that the probability that $J$ passes the GHZ test is given by (\ref{expforpassprob}).
We prove the proposition in $3$ cases.

\vskip0.2in 

\textbf{Case 1:} $\mathbf{c_{000} = c_{111} = 0}$.

Using Lemma~\ref{threevariablelemma2}, we have the following:

\begin{eqnarray}
1 - \epsilon \hskip0.05in = \hskip0.05in \frac{1}{2} & + &
\textnormal{Re} \left[ \overline{c_{110}} \cdot c_{001} \left( \frac{ -1 + \lambda \gamma +
\gamma \overline{\phi} + \overline{\phi } \lambda }{4} \right) \right] \\
\nonumber
& + &
\textnormal{Re} \left[ \overline{c_{101}} \cdot c_{010} \left( \frac{ -1 + \lambda \overline{ \gamma } +
\overline{ \gamma } \phi + \phi \lambda }{4} \right) \right] \\
\nonumber
& + &
\textnormal{Re} \left[ \overline{c_{100}} \cdot c_{011} \left( \frac{ -1 + \lambda \overline{ \gamma } +
\overline{ \gamma } \overline{ \phi } + \overline{ \phi } \lambda }{4} \right) \right] \\
\hskip0.05in \leq \hskip0.05in \frac{1}{2} & + &
\left|  c_{110} \right| \left| c_{001} \right| \left| \frac{ -1 + \lambda \gamma +
\gamma \overline{\phi} + \overline{\phi } \lambda }{4} \right| \\
\nonumber
& + &
\left| c_{101} \right| \left|  c_{010} \right| \left| \frac{ -1 + \lambda \overline{ \gamma } +
\overline{ \gamma } \phi + \phi \lambda }{4} \right| \\
\nonumber
& + &
\left| c_{100} \right| \left| c_{011} \right| \left| \frac{ -1 + \lambda \overline{ \gamma } +
\overline{ \gamma } \overline{ \phi } + \overline{ \phi } \lambda }{4} \right|  \\
\nonumber
\hskip0.05in \leq \hskip0.05in \frac{1}{2}  & + &
\left( \frac{ \left| c_{110}  \right|^2 +  \left|  c_{001} \right|^2 }{2} \right)  \cdot 
\frac{\sqrt{2}}{2} \\
\nonumber
& + &
\left( \frac{ \left|  c_{101}  \right|^2 +  \left|  c_{010} \right|^2 }{2} \right)  \cdot 
\frac{\sqrt{2}}{2}  \\
\nonumber
& + &
\left( \frac{ \left|  c_{100}  \right|^2 +  \left|  c_{011} \right|^2 }{2} \right)  \cdot \frac{\sqrt{2}}{2} \\
\hskip0.05in = \hskip0.05in \frac{1}{2} & + & \frac{\sqrt{2}}{4}
\end{eqnarray}
Therefore,
\begin{eqnarray}
\epsilon \geq 1 - \left( \frac{1}{2} + \frac{\sqrt{2}}{4} \right)
\end{eqnarray}
Since the quantity above is greater than $\frac{1}{7}$, condition
(\ref{prestateinequality}) is trivially satisfied.  

\hskip0.2in

\textbf{Case 2:} $\mathbf{c_{001} = c_{010} = c_{100} = c_{011} = c_{101} = c_{110} = 0}$.

We have
\begin{eqnarray*}
1 - \epsilon &  = & \frac{1}{2} +
\textnormal{Re} \left[ \overline{c_{111}} \cdot c_{000} \left( \frac{ -1 + \lambda \gamma +
\gamma \phi + \phi \lambda }{4} \right) \right] \\
& = & \textnormal{Re} \left[ \overline{c_{111}} \cdot c_{000} \left( \frac{-1}{4} \right) \right]
+ \textnormal{Re} \left[ \overline{c_{111}} \cdot c_{000} \left( \frac{ \lambda \gamma +
\gamma \phi + \phi \lambda }{4} \right) \right] \\
& \leq & \textnormal{Re} \left[ \overline{c_{111}} \cdot c_{000} \left( \frac{-1}{4} \right) \right]
+ \left|  c_{111} \right| \left| c_{000}  \right| \left| \frac{ \lambda \gamma +
\gamma \phi + \phi \lambda }{4} \right| \\
& \leq & \textnormal{Re} \left[ \overline{c_{111}} \cdot c_{000} \left( \frac{-1}{4} \right) \right]
+ \left( \frac{ \left|  c_{111} \right|^2 +  \left| c_{000}  \right|^2}{2} \right) \left| \frac{ \lambda \gamma +
\gamma \phi + \phi \lambda }{4} \right| \\
& \leq & \textnormal{Re} \left[ \overline{c_{111}} \cdot c_{000} \left( \frac{-1}{4} \right) \right]
+ \left( \frac{1}{2} \right) \left( \frac{3}{4} \right)
\end{eqnarray*}
which implies, by a linear manipulation,
\begin{eqnarray}
1 - 8 \epsilon  \leq ( -2 )  \cdot \Re \left[ \overline{c_{111}} c_{000} \right]
\end{eqnarray}
We have the following:
\begin{eqnarray}
2 - 8 \epsilon & \leq & 1 - 2 \cdot \Re \left[ \overline{c_{111}} c_{000} \right]
\end{eqnarray}
By an easy calcuation, the right side of this inequality is equal to 
$2 \left| \left< \alpha, g \right> \right|^2$.  Thus,
\begin{eqnarray}
2 - 8 \epsilon & \leq & 2 \left| \left< \alpha , g \right> \right|^2 \\
1 - 4 \epsilon & \leq & \left| \left< \alpha , g \right> \right|^2 \\
1 - 4 \epsilon & \leq & \left| \left< \alpha , g \right> \right| \\
1 - 7 \epsilon & \leq & \left| \left< \alpha , g \right> \right|.
\end{eqnarray}

\vskip0.2in

\textbf{Case 3: General case.}

Let
\begin{eqnarray}
\alpha = x \alpha' + y \alpha''
\end{eqnarray}
where $\alpha', \alpha''$ are unit vectors satisfying the
conditions of Case 1 and Case 2, respectively, and
$(x, y)$ is a unit-length vector in $\mathbb{R}^2$.
Let $J_{\alpha'}$ and $J_{\alpha''}$ denote the
device $J$ with the state $\alpha$ replaced by $\alpha'$
and $\alpha''$, respectively.  It is clear from expression (\ref{expforpassprob})
that
\begin{eqnarray}
\nonumber
\mathbf{P} \left( \textnormal{$J$ passes the GHZ test} \right)
& = & x^2 \cdot \mathbf{P} \left( \textnormal{$J_{\alpha'}$ passes
the GHZ test} \right) \\
\nonumber
& + &  y^2 \cdot \mathbf{P} \left( \textnormal{$J_{\alpha''}$ passes the
GHZ test} \right).
\end{eqnarray}
This case follows from cases 1 and 2 by an easy linearity argument.
\end{proof}

\begin{corollary}
\label{distancecorollary}
Let $J$ be the device from Proposition~\ref{msrmtapproxprop}.  There
exists a complex number $\zeta$ with $\left| \zeta \right| = 1$ such
that
\begin{eqnarray}
\left| \left| \alpha - \zeta g \right| \right|^2 \leq 14 \epsilon.
\end{eqnarray}
\end{corollary}

\begin{proof}
Choose $\zeta$ so that the inner product $\left< \alpha , \zeta g \right>$
is a nonnegative real number.  (Thus, $\left< \alpha , \zeta g \right>
= \left| \left< \alpha, g \right> \right|$.)  Then,
\begin{eqnarray}
\left| \left| \alpha - \zeta g \right| \right|^2 & = & \left< \alpha , \alpha \right> - 
\left< \alpha , \zeta g \right> - 
\left< \zeta g, \alpha \right> + 
\left< \zeta g , \zeta g \right>  \\
& = & 2 - 2 \left| \left< \alpha, g \right> \right| \\
& \leq & 2 - 2 ( 1 - 7 \epsilon ) \\
& = & 14 \epsilon.
\end{eqnarray}
\end{proof}

\subsection{The Post-Measurement State of an Approximate GHZ Device}

\label{postmsrmtsubsection}

Let $J$ be the $3$-part RE device from Proposition~\ref{msrmtapproxprop}.
Suppose, as usual, that $J$ is given an input string $I_1 I_2 I_3$ that is randomly
chosen from the set $\{ 000 , 011 , 101 , 110 \}$.  Consider the post-measurement
state of $J$ together with its input and output registers.
\[
\xymatrix{
I_1 \ar[d] & I_2 \ar[d] & I_3 \ar[d] \\
\framebox{ \framebox[0.5in]{$R_1$} }_{J_1} \ar[d] &
\framebox{ \framebox[0.5in]{$R_2$} }_{J_2} \ar[d] &
\framebox{ \framebox[0.5in]{$R_3$} }_{J_3} \ar[d] \\
O_1 & O_2 & O_3
}
\]
Recall that the pre-measurement state of $R = (R_1, R_2, R_3)$ is a pure state $\alpha$.
We can express the post-measurement state of $IOR$ as a pure state like so:
\begin{eqnarray}
\label{postmsrmtexp}
v^{(post)}_{IOR} = 
 \left[  \left( \frac{1}{4} \right) \cdot \sum_{(i_k) , (o_k)}
\left| i_1 i_2 i_3 \right> \otimes \left| o_1 o_2 o_3 \right> \otimes \right.
\hskip1in \\
\nonumber
\left. \left( 
\frac{\mathbb{I} + (-1)^{o_1} S_1^{(i_1)}}{2} \right) \otimes
\left( 
\frac{\mathbb{I} + (-1)^{o_2} S_2^{(i_2)}}{2} \right) \otimes
\left( 
\frac{\mathbb{I} + (-1)^{o_3} S_3^{(i_3)}}{2} \right)  \right] \alpha.
\end{eqnarray}
where the summation is taken over all input/output combinations $(i_k)$, $(o_k)$
which satisfy the GHZ test ((\ref{ghztest1})--(\ref{ghztest2})).

Let
\begin{eqnarray}
g = \frac{1}{\sqrt{2}} \left( \left| 000 \right> - \left| 111 \right> \right)
\end{eqnarray}
and
\begin{eqnarray}
\sigma^{(0)} = \left[ \begin{array}{ccc} 0 & 1 \\ 1 & 0 \end{array} \right] & & 
\sigma^{(1)} = \left[ \begin{array}{ccc} 0 & i \\ -i & 0 \end{array} \right]
\end{eqnarray}
Define a new state, denoted $v^{(ideal)}_{IOR}$, by expression 
(\ref{postmsrmtexp}) with the operators $S^{(i)}_k$ replaced by $\sigma^{(i)}$
and the state $\alpha$ replaced by $g$:
\begin{eqnarray}
\label{idealexp}
v^{(ideal)}_{IOR} = 
 \left[ \left( \frac{1}{4} \right) \cdot \sum_{(i_k) , (o_k)}
\left| i_1 i_2 i_3 \right> \otimes \left| o_1 o_2 o_3 \right> \otimes \right.
\hskip1in \\
\nonumber
\left. \left( 
\frac{\mathbb{I} + (-1)^{o_1} \sigma^{(i_1)}}{2} \right) \otimes
\left( 
\frac{\mathbb{I} + (-1)^{o_2} \sigma^{(i_1)}}{2} \right) \otimes
\left( 
\frac{\mathbb{I} + (-1)^{o_3} \sigma^{(i_1)}}{2} \right)  \right] g.
\end{eqnarray}
This is the post-measurement state of the ideal GHZ device
from subsection~\ref{characterizationsubsection}.

Let us compare expressions (\ref{postmsrmtexp}) and
(\ref{idealexp}).
From Proposition~\ref{msrmtapproxprop}, we have
\begin{eqnarray}
\left| \left|  S^{(i_k)}_k - \sigma^{(i_k)} \right| \right|_2^2
\leq 32 \epsilon.
\end{eqnarray}
Therefore,
\begin{eqnarray}
\left| \left|  S^{(i_k)}_k - \sigma^{(i_k)} \right| \right|_\infty
& \leq & \left| \left|  S^{(i_k)}_k - \sigma^{(i_k)} \right| \right|_2 \\
& \leq &
\sqrt{32 \epsilon}.
\end{eqnarray}
The individual projection operators from 
(\ref{postmsrmtexp}) and
(\ref{idealexp}) therefore satisfy
\begin{eqnarray}
\left| \left|  \left( 
\frac{\mathbb{I} + (-1)^{o_k} S_1^{(i_k)}}{2} \right) - \left( 
\frac{\mathbb{I} + (-1)^{o_k} \sigma^{(i_k)}}{2} \right) \right| \right|_\infty
\leq \frac{\sqrt{ 32 \epsilon }}{2} = \sqrt{8 \epsilon}
\end{eqnarray}
Let $M$ and $M^{(ideal)}$ be the operators enclosed in
brackets in (\ref{postmsrmtexp}) and (\ref{idealexp}), respectively.  Applying
Lemma~\ref{infinitynormlemma}, we find that
\begin{eqnarray}
\left| \left| M - M^{(ideal)} \right| \right|_\infty 
\leq 3 \cdot \sqrt{8 \epsilon}.
\end{eqnarray}

By Corollary~\ref{distancecorollary}, there is a unit-length
complex number $\zeta$ such that
\begin{eqnarray}
\left| \left| \alpha - \zeta  g \right| \right|^2 \leq 14 \epsilon.
\end{eqnarray}
Therefore:
\begin{eqnarray*}
\left| \left| v^{(post)}_{IOR} - \zeta \cdot v^{(ideal)}_{IOR} \right| \right|
& = &
\left| \left| M \alpha - M^{(ideal)} (\zeta g) \right| \right| \\
& \leq & \left| \left| M \alpha - M^{(ideal)} \alpha \right| \right| 
+ \left| \left| M^{(ideal)} \alpha - M^{(ideal)} (\zeta g ) \right| \right| \\
& \leq & \left| \left| M - M^{(ideal)} \right| \right|_\infty
+ \left| \left| M \right| \right|_\infty \left| \left| \alpha
- \zeta g \right| \right| \\
& \leq & 3 \cdot \sqrt{8 \epsilon} + \sqrt{14 \epsilon}.
\end{eqnarray*}
Squaring, we find
\begin{eqnarray}
\left| \left| v^{(post)}_{IOR} - \zeta \cdot v^{(ideal)}_{IOR} \right| \right|^2
& \leq  & \left( 3 \sqrt{8} + \sqrt{14} \right)^2 \epsilon \\
& < & 150 \epsilon.
\end{eqnarray}

We state the above inequality as a proposition.

\begin{proposition}
\label{222approxstateprop}
Let $J$ be the device from Proposition~\ref{msrmtapproxprop}.
Suppose that
\begin{eqnarray}
\mathbf{P} \left( \textnormal{$J$ passes the GHZ
test} \right) = 1 - \epsilon.
\end{eqnarray}
Then, there exists a unit-length
complex number $\zeta$ such that
\begin{eqnarray}
\label{150ineq}
\left| \left| v^{(post)}_{IOR} - \zeta \cdot v^{(ideal)}_{IOR} \right| \right|^2
& \leq & 150 \epsilon
\end{eqnarray}
where $v^{(ideal)}_{IOR}$ denotes the post-measurement
state of the ideal GHZ device.  
\end{proposition}

\section{The GHZ Paradox: General Case}

Now we will broaden our focus and consider $3$-part
RE devices that contain arbitrarily large quantum systems.

\subsection{The Post-Measurement State of
an Approximate GHZ Device: General Case}

As in subsection~\ref{postmsrmtsubsection}, if
$R$ is a $3$-qubit system, then we write
$v^{(ideal)}_{IOR}$ for the post-measurement
state of an ideal GHZ device on $R$.  Also,
we will write $\Gamma^{(ideal)}_{IOR}$ for
the corresponding density operator.

\begin{theorem}
\label{closenesstheorem}
Let $D'$ be a $3$-part RE device that is in the canonical form described
in Proposition~\ref{canonicalform}.  Let $F$ be a quantum system
which is entangled with $D'$ in such a way that
the state $\Gamma_{Q' F}$ is pure.
Let $I$ and $O$ denote the input and output registers of $D'$.
\[
\xymatrix{
I_1 \ar[d] & I_2 \ar[d] & I_3 \ar[d] \\
\framebox{ \framebox[0.5in]{$Q'_1$} }_{D'_1} \ar[d] &
\framebox{ \framebox[0.5in]{$Q'_2$} }_{D'_2} \ar[d] &
\framebox{ \framebox[0.5in]{$Q'_3$} }_{D'_3} \ar[d] &
\framebox[0.5in]{$F$}  \\
O_1 & O_2 & O_3
}
\]
Suppose that
\begin{eqnarray}
\mathbf{P} \left( \textnormal{$D'$ satisfies the GHZ
paradox} \right) = 1 - \epsilon.
\end{eqnarray}
Then, there exists a unit vector
\begin{eqnarray}
w \in \mathcal{A} \otimes \mathcal{B} \otimes \mathcal{C} \otimes \mathcal{F}
\end{eqnarray}
such that the inequality
\begin{eqnarray}
\label{vectordifference150}
\left| \left| v^{(post)}_{IOQ'F} - w \otimes v^{(ideal)}_{IOR} \right| \right|^2
\leq 150 \epsilon
\end{eqnarray}
holds.\footnote{Note that in expression~(\ref{vectordifference150}), we have treated
the vector spaces $\mathcal{R}_j$ from the statement of
Proposition~\ref{canonicalform} as the state spaces of a $3$-qubit system.
The expression $v^{(ideal)}_{IOR}$ represents the post-measurement
state of an ideal GHZ device on this system.}
\end{theorem}

\begin{proof}
Choose an orthonormal basis $\{ f_n \}$ for $F$.
Let us write the pre-measurement state $v^{(pre)}_{Q'F}$ as
\begin{eqnarray}
v^{(pre)}_{Q'F} = \sum_{klmn} \left( a_k \otimes b_l \otimes c_m \otimes
f_n \right) \otimes \alpha_{klmn}.
\end{eqnarray}
with $\alpha_{klmn} \in \mathcal{R}_1 \otimes \mathcal{R}_2 \otimes
\mathcal{R}_3$.

For every $4$-tuple $(k, l, m, n)$ which is such that
$\alpha_{klmn} \neq 0$, define $J_{klmn}$ to
be the $3$-part RE device on the system
$(R_1, R_2, R_3)$ whose pre-measurement state
is
\begin{eqnarray}
\frac{ \alpha_{klmn} }{\left| \alpha_{klmn} \right|}
\end{eqnarray}
and whose measurement operators are given by
\begin{eqnarray}
\begin{array}{ccc}
S_1^{(0)} =  \left[ \begin{array}{cc} 0 & 1 \\ 1 & 0 \end{array} \right]  & 
S_2^{(0)} = \left[ \begin{array}{cc} 0 & 1 \\ 1 & 0 \end{array} \right] & 
S_3^{(0)} = \left[ \begin{array}{cc} 0 & 1 \\ 1 & 0 \end{array} \right] \\ \\
S_1^{(1)} =  \left[ \begin{array}{cc} 0 & \lambda_k \\ \overline{\lambda_k} & 0 \end{array} \right]  & 
S_2^{(1)} = \left[ \begin{array}{cc} 0 & \gamma_l \\ \overline{\gamma_l} & 0 \end{array} \right] & 
S_3^{(1)} = \left[ \begin{array}{cc} 0 & \phi_m \\ \overline{\phi_m} & 0 \end{array} \right] \\ 
\end{array}
\end{eqnarray}
For each device $J_{klmn}$, choose a unit-length complex number
$\zeta_{klmn}$ so as to minimize the Euclidean distance between
$\zeta_{klmn} \cdot v^{(ideal)}_{IOR}$ and the post-measurement
state of $J_{klmn}$.  Let
\begin{eqnarray}
w = \sum_{klmn}  \left( a_k \otimes b_l \otimes c_m \otimes
f_n \right) \otimes \zeta_{klmn} \left| \alpha_{klmn} \right|.
\end{eqnarray}

Note that
\begin{eqnarray*}
\mathbf{P} \left( \textnormal{$D'$ passes the GHZ test} \right)
& = & \sum_{klmn} \left| \alpha_{klmn} \right|^2 \mathbf{P}
\left( \textnormal{$J_{klmn}$ passes the GHZ test} \right)
\end{eqnarray*}
Likewise, we can express the quantity
\begin{eqnarray}
\left| \left|   v^{(post)}_{IOQ'F} - w \otimes v^{(ideal)}_{IOR}
\right| \right|^2
\end{eqnarray}
as a weighted sum of the squares of the distances between the vectors
$\zeta_{klmn}~\cdot~v^{(ideal)}_{IOR}$ and the post-measurement
states of $J_{klmn}$.  Inequality (\ref{vectordifference150})
thus follows from Proposition~\ref{222approxstateprop} by linearity.
\end{proof}

Let us consider how Theorem~\ref{closenesstheorem} may be strengthened.
We wish to cover 
the case in which the stated device $D'$
is entangled with another quantum system $F_1$, but the
state $\Gamma^{(pre)}_{Q'F_1}$ is not pure.
In this case, we can simply introduce a second entangled quantum
system $F_2$ such that $\Gamma^{(pre)}_{Q'F_1F_2}$ is pure.  By the
theorem, we find that for some vector $w$ in 
the tensor product of $\mathcal{A}$, $\mathcal{B}$, $\mathcal{C}$,
$\mathcal{F}_1$, and $\mathcal{F}_2$,
\begin{eqnarray}
\left| \left| v^{(post)}_{IOQ'F_1F_2} - w \otimes v^{(ideal)}_{IOR}
 \right| \right|^2 \leq 150 \epsilon.
\end{eqnarray}
We would like a similar expression with the system $F_2$ omitted.
Since we cannot use vector-length to measure distance in this case,
we use the trace-norm instead.  
\begin{eqnarray*}
\left| \left| \Gamma^{(post)}_{IOQ' F_1} - 
\Tr_{\mathcal{F}_2} \left( w w^\dagger \right) \otimes
\Gamma^{(ideal)}_{IOR} \right| \right|^2_1 & \leq &
\left| \left| \Gamma^{(post)}_{IOQ' F_1 F_2} - 
 w w^\dagger \otimes
\Gamma^{(ideal)}_{IOR} \right| \right|^2_1 \\
& \leq & 
\left( 2 \left| \left| v^{(post)}_{IOQ'F_1F_2} - w \otimes v^{(ideal)}_{IOR}
 \right| \right| \right)^2 \\
& \leq & 600 \epsilon.
\end{eqnarray*}

We have the following alternate statement of the theorem.

\begin{theorem}
Let $D'$ be a device which is in the canonical form described in Proposition~\ref{canonicalform}.
Let $F_1$ be a quantum system which may be entangled (in a possibly mixed
state) with $D'$.
Suppose that $D'$ passes the GHZ test with probability $1 - \epsilon$.
Then, there exists a density operator
$\Phi_{ABCF_1}$ on $\mathcal{A} \otimes \mathcal{B}
\otimes \mathcal{C} \otimes \mathcal{F}_1$ such that
\begin{eqnarray}
\label{corollaryinequality}
\left| \left| \Gamma^{(post)}_{IOQ'F_1} - \Phi_{ABCF_1} \otimes 
\Gamma^{(ideal)}_{IOR} \right| \right|_1^2 \leq 600 \epsilon.  
\end{eqnarray}
\end{theorem}

Secondly, we would like to have a version of the theorem
which applies to devices that are not in canonical form. 
To simplify matters, let us drop the quantum system $Q'$
and all of its subsystems from (\ref{corollaryinequality}).
We have:
\begin{eqnarray}
\left| \left| \Gamma^{(post)}_{IOF_1} - \Phi_{F_1} \otimes 
\Gamma^{(ideal)}_{IO} \right| \right|_1^2 \leq 600 \epsilon
\end{eqnarray}
Since a unitary embedding of one entangled $3$-part RE device
into another obviously has no effect on the quantity in this inequality,
we have the following by Proposition~\ref{canonicalform}.

\begin{theorem}
Let $D$ be a $3$-part RE device (not necessarily in canonical form)
which passes the GHZ test with probability $1 - \epsilon$.
Let $E$ be a quantum system which may be entangled with $D$.
Then, there exists a density operator $\Phi_E$ on $\mathcal{E}$
such that
\begin{eqnarray}
\label{modifiedineq}
\left| \left| \Gamma^{(post)}_{IOE} - \Phi_E \otimes 
\Gamma^{(ideal)}_{IO} \right| \right|_1^2 \leq 600 \epsilon. 
\end{eqnarray}
\end{theorem}

Finally, note that inequality (\ref{modifiedineq}) above implies
\begin{eqnarray}
\left| \left| \Gamma_E - \Phi_E \right| \right|^2
\leq 600 \epsilon.
\end{eqnarray}
Therefore,
\begin{eqnarray}
\left| \left| \Gamma_E \otimes 
\Gamma^{(ideal)}_{IO} -
\Phi_E \otimes 
\Gamma^{(ideal)}_{IO} \right| \right|^2_1 \leq 600 \epsilon.
\end{eqnarray}
Applying the triangle inequality for $\left| \left| \cdot \right| \right|_1$, we find
that
\begin{eqnarray}
\left| \left| \Gamma^{(post)}_{IOE} - \Gamma_E \otimes 
\Gamma^{(ideal)}_{IO} \right| \right|_1^2 \leq 2400 \epsilon. 
\end{eqnarray}
We therefore have the following final version of our theorem.
\begin{theorem}
\label{finaltheorem}
Let $D$ be a $3$-part RE device
which passes the GHZ test with probability $1 - \epsilon$.
Let $E$ be a quantum system which may be entangled with $D$.
Then,
\begin{eqnarray}
\left| \left| \Gamma^{(post)}_{IOE} - \Gamma_E \otimes 
\Gamma^{(ideal)}_{IO} \right| \right|_1^2 \leq 2400 \epsilon. 
\end{eqnarray}
\end{theorem}